\documentclass[english,a4paper,12pt]{article}
\usepackage[utf8]{inputenc}
\usepackage[T1]{fontenc}
\usepackage{titling}
\usepackage[colorlinks,citecolor=red,urlcolor=blue,bookmarks=false,hypertexnames=true]{hyperref} 

\usepackage{amsmath, amssymb}
\usepackage{pgf}
\usepackage{amsfonts}
\usepackage{dsfont}
\usepackage{mathrsfs}
\usepackage{mathabx}
\usepackage{stmaryrd}
\usepackage{makeidx}
\usepackage{amsbsy}
\usepackage{amsthm}
\usepackage{color}
\usepackage{fullpage}
\usepackage{enumitem}
\usepackage[vcentermath]{youngtab}
\usepackage{marginnote} 
\usepackage{mdframed}	
\newmdenv[topline=false, bottomline=false, skipabove=\topsep, skipbelow=\topsep]{siderules}
\newtheorem{theorem}{Theorem}
\newtheorem{corollary}{Corollary}

\newtheorem{proposition}{Proposition}
\newtheorem{lemma}{Lemma}

\newtheorem{definition}{Definition}

\def\sD{\mathscr{D}}

\def\sH{\mathscr{H}}

\def\sK{\mathscr{K}}
\def\sL{\mathscr{L}}

\def\sP{\mathscr{P}}

\def\sS{\mathscr{S}}

\def\bN{{\mathbb N}}

\def\bS{{\mathbb S}}

\def\cA{{\mathfrak A}}
\def\cB{{\mathfrak B}}

\def\cM{{\mathfrak M}}

\newcommand{\ca}[1]{{\cal #1}}
\newcommand{\ben}{\begin{equation}}
\newcommand{\een}{\end{equation}}
\def\bena{\begin{eqnarray}}
\def\eena{\end{eqnarray}}

\def\cA{{\ca A}}
\def\cB{{\ca B}}
\def\cC{{\ca C}}

\def\cM{{\ca M}}

\renewcommand{\H}{\mathscr{H}}

\def\1{{\mathds{1}}}

\newcommand{\dd}{{\rm d}}
\newcommand{\tr}{\operatorname{Tr}}

\renewcommand{\log}{\operatorname{ln}}

\renewcommand{\epsilon}{\varepsilon}

\newcommand{\RR}{\mathbb{R}}
\newcommand{\CC}{\mathbb{C}}

\renewcommand{\Re}{{\rm Re}}

\begin{document}
\title{Trace- and improved data processing inequalities for von Neumann algebras 
}

	\author{Stefan Hollands$^{1}$\thanks{\tt stefan.hollands@uni-leipzig.de}\\
	{\it $^2$ ITP, Universit\" at Leipzig and MPI-MiS Leipzig}
	}

\date{\today}
	
\maketitle

\begin{abstract}
We prove a version of the data-processing inequality for the relative entropy for 
general von Neumann algebras with an explicit lower bound involving the measured relative entropy. 
The inequality, which generalizes previous work by Sutter et al. on 
finite dimensional density matrices, yields a bound how well a quantum state can be recovered 
after it has been passed through a channel. The natural applications of our results are in
quantum field theory where the von Neumann algebras are known to be of type III. Along the way we generalize 
various multi-trace inequalities to general von Neumann algebras.
\end{abstract}

\section{Introduction}

The relative entropy $S(\rho | \sigma) = \tr (\rho \log \rho-\rho \log \sigma)$ is an important operationally defined measure 
for the distinguishability of two statistical operators $\rho,\sigma$. A fundamental property of $S$ is that 
\ben\label{dpi}
S(\rho | \sigma) - S( T(\rho) | T(\sigma)) \ge 0
\een
for a quantum channel $ T$, i.e. completely positive linear trace preserving map\footnote{In the body of the paper, we use the slightly different notation 
$\tilde T$ for the action of a channel on a density matrix (Schr\" odinger picture), while $T$ denotes the dual action (Heisenberg picture) of the channel on the observables.}. The above difference represents the loss of distinguishability between 
$\sigma,\rho$ if these are passed through the channel $T$.

An important general question that can be abstracted from 
concrete settings such as quantum communication or quantum error correction is to what extent the action of a quantum 
channel can be reversed, i.e. to what extent it may be possible to recover $\rho$ from $ T(\rho)$. It was understood already a long 
time ago by Petz that the question of recoverability is intimately linked to the case of saturation of the data processing inequality (DPI) \eqref{dpi}, see e.g. 
\cite{Petz1993}. As was understood by \cite{Fawzi} -- and has subsequently been generalized in various was by \cite{Junge,Sutter2,Fawzi,Berta,Carlen,Jencova,Sutter,Wilde} -- explicit lower bounds in the DPI or related information theoretic inequalities can provide information how well a channel may be reversed if the inequality is e.g. nearly saturated. 

The current best result in this direction appears to be that by Sutter, Berta, and Tomamichel \cite{Sutter2}. It provides an explicit recovery channel, such that the recovered state is close to the original state $\rho$ in a suitable information theoretic measure provided the difference in the DPI is also small. The recovery channel $\alpha_{\sigma,T}$ is called ``explicit'' because it is given by a concrete expression involving only reference state $\sigma$ and $T$ (not the state $\rho$
that is to be recovered), and always perfectly recovers $\sigma$, i.e. $\alpha_{\sigma,T}(T(\sigma))=\sigma$. In fact, it is closely related -- though not precisely equal -- to the channel originally proposed by Petz \cite{Petz1,Petz2,Petz3,Petz1993}. 

The above mentioned works (though not \cite{Petz1,Petz2,Petz3,Petz1993}) establish their results only for 
type I von Neumann algebras -- in particular \cite{Sutter2} assumes a finite-dimensional Hilbert space. While this is well-motivated by applications in quantum computing, there are cases of interest when the algebras are not of this type. A notable example of this are quantum field theoretic applications
related to the ``quantum null energy condition'' (see e.g. \cite{Faulkner}) where the algebras are of type III \cite{Buchholz,haag_2}. 
With this application in mind we proved in \cite{hollands1} a generalization of \cite{Junge} in the case when the channel $T$ corresponds to 
an inclusion of general von Neumann algebras. This result has been generalized to arbitrary 2-positive channels $T$ in \cite{hollands2}, where the following improved DPI has been demonstrated:
\ben
\label{eq:Jensen0}
S(\rho | \sigma) - S(T(\rho) | T(\sigma))
\ge \frac{1-s}{s}  \int_{\RR} dt \, \beta_0(t) \, D_{s} ( \alpha_{\sigma,T}^t(T(\rho)) | \rho).
\een
Here, $s \in [1/2,1)$ and $D_s$ are the so-called ``sandwiched Renyi entropies'' \cite{mueller,wilde3}, which for $s=1/2$ become the negative log squared fidelity.
$\beta_0(t) dt$ is a certain explicit probability density and $\alpha^t_{\eta,T}$ is an explicit 1-parameter family of recovery channels that is a disintegration of $\alpha_{\eta,T}$ in the sense $\int dt \, \beta_0(t) \, \alpha_{\eta,T}^t = \alpha_{\eta,T}$. Using convexity of $D_s$ and Jensen's inequality, 
the bound implies
\ben
\label{eq:Jensen}
S(\rho | \sigma) - S(T(\rho) | T(\sigma))
\ge \frac{1-s}{s}  D_{s} ( \alpha_{\sigma,T}(T(\rho)) | \rho).
\een
A qualitatively similar result has been proved for general von Neumann algebras 
by Junge and LaRacuente \cite{Rac}. In their result, the sandwiched Renyi entropies are now replaced by some
other information theoretic quantity with an operational meaning. Both \cite{hollands2,Rac} lead to the same inequality for $s=1/2$. For 
type I algebras and $s=1/2$ \eqref{eq:Jensen0} is the result by \cite{Junge}, but the relation for 
general $s$ is unclear  to the author. We also mention recent results by Gao and Wilde \cite{Gao} of a roughly similar flavor but different emphasis, which apply to von Neumann algebras with a trace though not type III. 

In the present paper, we provide a generalization of \cite{Sutter2} to arbitrary (sigma-finite) von Neumann algebras. This version of the improved DPI is qualitatively similar to \eqref{eq:Jensen}. The definition of the recovery channel is in fact identical to that in \eqref{eq:Jensen}, but we have yet another information theoretic quantity on the 
right side, namely (thm. \ref{thm:main})
\ben
\label{eq:Jensen1}
S(\rho | \sigma) - S(T(\rho) | T(\sigma))
\ge S_{\rm meas} ( \alpha_{\sigma,T}(T(\rho)) | \rho).
\een 
 Here, $S_{\rm meas}$ is the ``measured relative entropy'', defined as the maximum possible value of the relative entropy restricted to a 
 commutative subalgebra. We show below (prop. \ref{FS}) that for $s=1/2$, this inequality is sharper than \eqref{eq:Jensen} -- though not 
 in general the inequality \eqref{eq:Jensen0} with the integral outside -- for all $\rho,\sigma$. 
A conceptual advantage of \eqref{eq:Jensen1} over both \eqref{eq:Jensen0} and \eqref{eq:Jensen} (and likewise to the inequalities proven in \cite{Rac}) is that it is saturated in the commutative case, as noted already by  \cite{Sutter2}. So in this respect \eqref{eq:Jensen1} is sharp unlike its predecessors.
 
Our proof technique is similar in several respects to that in \cite{Sutter2} and related antecedents such as \cite{Junge} in that we also use interpolation arguments for $L_p$-spaces. However, there are also some key differences requiring technical modifications: For instance, the operators $\log \rho$ or 
$\log \sigma$ no longer exist for general von Neumann algebras or the use of ordinary $L_p$ (Schatten)-spaces is prohibited since a general von Neumann algebra does not have a trace. As in our previous papers \cite{hollands1,hollands2} -- referred to as papers I,II -- 
our solution to the first problem is to work entirely 
with Araki's relative modular operator, the log of which can roughly be viewed as a {\em difference} between $\log \rho$ and $\log \sigma$. 
Likewise, as in \cite{hollands1,hollands2}, our solution to the second problem is to work with the Araki-Masuda non-commutative $L_p$-spaces \cite{AM} which are very closely related to the sandwiched relative Renyi entropies\footnote{\cite{Rac} use a somewhat different approach to $L_p$
spaces to circumvent the absence of a tracial state in the general von Neumann algebra setting. Their approach appears to us less natural for the purposes of this paper.}. For these norms, we require a complex interpolation theory, see lem. \ref{lem:hirsch}, which generalizes a result in \cite{hollands1}. This result is then applied to a specially constructed analytic family of vectors and combined with certain cutoff-techiques for appropriately extended domains of analyticity in a similar way as in \cite{hollands1}. However, in \cite{hollands1,hollands2}, such cutoff techniques were needed to control the limit of the 
Araki-Masuda norms as $p \to 2$, whereas in the present paper, it is the limit $p \to \infty$ which is relevant. The regularization is necessary here to apply the powerful technique of bounded perturbations of normal states of a von Neumann algebra, and a (somewhat modified) version of the Lie-Trotter product formula for von Neumann algebras \cite{Araki4}. These ideas go beyond \cite{hollands1,hollands2} and also yield various new ``trace'' inequalities for von Neumann algebras which could be of independent interest. 

This paper is organized as follows. In sec. \ref{sec:1} we review some prerequisite notions from the theory of von Neumann algebras. 
In sec. \ref{sec:2} we establish an interpolation theorem for the Araki-Masuda $L_p$-norms, which we apply in sec. \ref{sec:3} to 
obtain generalizations of various known mutli-trace inequalities to von Neumann algebras. In sec. \ref{sec:4} we establish our main result, thm. \ref{thm:main}. The definition of the $L_p$-norms is relegated to the appendix.

\section{Von Neumann algebras and modular theory}
\label{sec:1}

Let $\mathcal{A} = M_{n}({\mathbb C})$. The fundamental representation of this algebra is on $\CC^n$, but one can also 
work in the ``standard'' Hilbert space ($\sH \simeq M_{n}({\mathbb C}) \simeq
{\mathbb C}^{n} \otimes  \mathbb{C}^n$). 
Vectors $|\zeta\rangle$ in $\sH$ are thus identified with matrices $\zeta \in M_n(\CC)$.
$\sH \simeq M_{n}({\mathbb C})$ is both a left and right module for $\cA$, 
\begin{equation}
l(a) \left| \zeta \right> = \left| a \zeta \right>\, \qquad r(b) \left| \zeta \right> = \left| \zeta b \right>, 
\end{equation}
and the inner product on $\sH$ is the Hilbert-Schmidt inner product $\langle \zeta_1 | \zeta_2\rangle = \tr(\zeta_1^* \zeta_2)$. 
A mixed state, represented by a density matrix $\omega$, gives rise to a linear functional on $\cA$ by
\ben
\omega(a) = \tr (\omega a),
\een
where the functional and the state is denoted by the same symbol. These linear functionals 
are alternatively characterized by the property $\omega(a^*a) \ge 0 , \omega(1)=1$.

A ($\sigma$-finite) von Neumann algebra in standard form $\cM$ is an ultra-weakly closed linear subspace of the bounded operators on a 
Hilbert space $\sH$. $\cM$ should contain $1$, be closed under products and the $*$-operation
should have a cyclic and separating vector $|\psi\rangle \in \sH$. Cyclic and separating 
means that $\cM |\psi\rangle$ is dense in $\sH$ and $m|\psi\rangle = 0$ implies $m=0$. 
In the matrix example, $\psi$ should therefore be invertible. The set of ultra-weakly continuous positive linear functionals
(thus satisfying $\omega(a^*a) \ge 0 , \omega(1)=1$) is called $\sS(\cM)$. For a detailed account of von Neumann 
algebras see \cite{Takesaki}.

Associated with a  von Neumann algebra in standard form\footnote{
More precisely, a standard form is actually 
defined by the combined structure $(\cM, \sH, \sP^\sharp_\cM, J)$, which can be recovered if we have a 
cyclic and separating vector.
} 
is a convex cone $\sP^\sharp_\cM$ 
and an anti-linear involution $J$, called ``modular conjugation'' leaving this cone invariant. A possible choice of this non-unique 
``natural cone'' for $\cA=M_n(\CC)$ is 
 the subset of positive semi-definite matrices in $\sH$, and in this case, $J \big| \zeta \big> = \big| \zeta^* \big>$. 
 A general property of $J$ which is easily verified in this example is that $J\cM J=\cM'$, the latter meaning the commutant of $\cM$ on $\sH$.
Given vectors $|\psi\rangle, |\eta\rangle, |\zeta\rangle \in \sH$ and $m \in \cM$, one defines following Araki \cite{Araki1} (see also 
app. C of \cite{AM} for many more details)
\ben
S_{\eta,\psi} \left( m |\psi\rangle + (1-\pi^{\cM'}(\psi)) |\zeta \rangle \right) = \pi^{\cM}(\psi) m^*|\eta\rangle. 
\een
Here $\pi^\cM(\psi) \in \cM$ is the orthogonal projection onto the closure of the subspace $\cM'|\psi\rangle$
and $\pi^{\cM'}(\psi) \in \cM'$ that onto the closure of $\cM |\psi\rangle$. The definition is consistent
because $m \pi^{\cM}(\psi) = 0$ if $m|\psi\rangle = 0$. 
One shows that $S_{\eta,\psi}$ is a closable 
operator and that if $|\psi\rangle \in \sP^\sharp_\cM$, then 
\ben
S_{\eta,\psi} = J\Delta_{\eta,\psi}^{1/2}, \quad 
 S_{\eta,\psi}^* \bar S_{\eta,\psi} = \Delta_{\eta,\psi}, 
\een
One calls the self-adjoint, non-negative operator 
$\Delta_{\eta,\psi}$ the ``relative modular operator''.  Its support is $\pi^\cM(\eta) \pi^{\cM'}(\psi)$ and 
complex powers $\Delta^z_{\eta,\psi}$ are understood as $0$ on the orthogonal complement of the support. 
The modular conjugation and relative modular operators of $\cA=M_n(\CC)$ with the above choice of  natural cone are:
\begin{equation}
\label{modobj}
 J \big| \zeta \big> = \big| \zeta^* \big> \, \qquad \Delta_{\eta,\psi} = l(\omega_\eta) r(\omega_\psi^{-1}),
\end{equation}
where we invert the density matrix $\omega_\psi$ on the range of $\pi^{\cM'}(\psi)$ which in the case at hand is 
the orthogonal projector onto the complement of the null space of $\omega_\psi$. 

For a general von Neumann algebra, every positive linear functional $\omega \in \sS(\cM)$ corresponds to one and only one
vector $|\xi_\omega\rangle$ in the natural cone $\sP_\cM^\natural$ such that $\omega(a) = \langle \xi_\omega | a \xi_\omega\rangle$. Vice versa, any 
vector $|\psi\rangle$ (in the natural cone or not) gives rise to a linear functional 
\ben
\omega_\psi(a) = \langle \psi | a \psi \rangle, \quad \text{for all $a \in \cA$.}
\een
For $\cA = M_n(\CC)$, this linear functional is identified with the density matrix $\omega_\psi = \psi \psi^*$ and the natural cone vectors correspond to the unique positive square root of the corresponding density matrix, now thought of as pure states in the standard Hilbert space. So the vector representative 
of a density matrix $\omega$ in the natural cone is $|\xi_\omega\rangle = |\omega^{1/2}\rangle$. An important fact used implicitly in several places below
is that if two linear functionals are close in norm, then the vectors in the natural cone are as well, and vice versa:
\ben\label{CS}
\| \xi_\psi - \xi_\eta \|^2 \le \| \omega_\eta - \omega_\xi\| \le \| \xi_\psi + \xi_\eta \| \,  \| \xi_\psi - \xi_\eta \| , 
\een
where the norm of a linear functional is $\|\omega\| = \sup\{ |\omega(m)| : m \in \cM, \|m\| =1 \}$. In the case $\cA=M_n(\CC)$, the latter norm 
is $\| \omega \| = \tr |\omega|$, so the above relation expresses the Powers-St\" ormer inequality between the trace norm and the Hilbert-Schmidt norm.

Let us finish this briefest of introduction to von Neumann
algebras by summarizing (again) some of our 

\medskip
\noindent
{\bf Notations and conventions:} Calligraphic letters $\cA, \cM, \dots$ denote von Neumann algebras, always assumed $\sigma$-finite. Calligraphic letters $\sH, \sK, \dots$ denote complex Hilbert spaces, always assumed to be separable. $\sS(\cM)$ denotes the set of all ultra-weakly continuous, positive, normalized linear functionals on $\cM$ (``states''), which 
are in one-to-one correspondence with density matrices if $\cA = M_n(\CC)$. $\cM_+$ is the subset of all non-negative self-adjoint operators in $\cM$
and $\cM_{\rm s.a.}$ the subset of all  self-adjoint elements of the von Neumann algebra $\cM$.
We use the physicist's ``ket''-notation $|\psi\rangle$ for vectors in a Hilbert space. 
The scalar product is written as 
\ben
(|\psi\rangle, |\psi'\rangle)_{\sH} =: \langle \psi | \psi' \rangle
\een
and is anti-linear in the first entry. The norm of a vector is written simply as
$\| |\psi\rangle \| =: \| \psi \|$. The action of a linear operator $T$ on a ket is sometimes written as $T|\phi\rangle = |T\phi\rangle$. 
In this spirit, the norm of a bounded linear operator $T$ on $\sH$ is written as $\|T\|= \sup_{|\psi\rangle: \|\psi\|=1} \|T\psi\|$. 

\section{Interpolation of non-commutative $L_p$ norms}
\label{sec:2}

For the algebra $\cA = M_n(\CC)$ the standard Hilbert space $\sH \cong M_n(\CC)$ on which $\cA$ acts by left multiplication
can be equipped with various norms. We have already mentioned that the 2-norm (or Hilbert-Schmidt norm)
\ben
\|\zeta\|_2 = (\tr \zeta \zeta^*)^{1/2},
\een
actually defines the Hilbert space norm on $\sH$ (so the subscript ``2'' is generally omitted). For $p>0$, one can generalize this to 
\ben
\|\zeta\|_p = [\tr (\zeta \zeta^*)^{p/2}]^{1/p}. 
\een
Given a faithful vector $|\psi\rangle \in \sH$ with associated linear functional $\omega_\psi(a) = \langle \psi | a \psi\rangle = \tr (a\omega_\psi)$ (Hilbert Schmidt inner product), one can also define the yet more general norms:
 \ben
\| \zeta \|_{p,\psi}= [ \tr(\zeta \omega_{\psi}^{2/p-1} \zeta^*)^{p/2}]^{1/p}.
\een
The faithful condition is relevant for $p>2$ as it ensures that $\omega_\psi$ is invertible.
The generalized $L_p$-norms $\| \zeta \|_{p,\psi}$ evidently reduce to usual $L_p$-norms if $\omega_\psi(a) = \tr(a)/n$ is the tracial state.
A general von Neumann algebra $\cM$ in standard 
form need not have such a tracial state, but Araki and Masuda \cite{AM} have shown that one can still define the above ``non-commuting $L_p$-norms'' for $p \ge 1$ using the relative modular operators based on $|\psi\rangle$, see also \cite{Jencova,JencovaLp1,Berta2}. Their basic definitions are recalled for convenience in the appendix. 
The following interpolation result for the Araki-Masuda $L_p$-norms is one of the main workhorses of this article.

\begin{lemma}
\label{lem:hirsch}
Let $|G(z)\rangle$ be a $\sH$-valued holomorphic function on the strip $\bS_{1/2}=\{0<{\rm Re}z<1/2\}$
that is uniformly bounded in the closure, and let $|\psi\rangle \in \sH$ a state of a
$\sigma$-finite von Neumann algebra $\mathcal M$ in standard form acting on $\sH$. 
For $0<\theta<1/2$, $p_0,p_1 \in [1,2]$ or $p_0,p_1 \in [2,\infty]$, let
\ben
\frac{1}{p_\theta} = \frac{1-2\theta}{p_0} + \frac{2\theta}{p_1}.
\een
Then 
\begin{align}
\label{himp}
& \ln \left\| G(\theta)\right\|_{p_\theta, \psi} \\
  \leq  & \int_{-\infty}^{\infty} d t 
  \left(
 (1-2 \theta)  \alpha_\theta(t) \ln \left\| G(it) \right\|_{p_0, \psi} + (2\theta)  \beta_\theta(t) 
  \ln  \left\| G(1/2+it) \right\|_{p_1, \psi} \right),
  \nonumber
\end{align}
where
\begin{equation}
\alpha_\theta(t) = \frac{ \sin(2\pi\theta)}{(1-2\theta)(\cosh(2\pi t ) - \cos(2\pi \theta)) }\,,
\qquad \beta_\theta(t) = \frac{ \sin(2\pi\theta)}{2 \theta(\cosh(2\pi t ) + \cos(2\pi \theta)) }.
\end{equation}
\end{lemma}

\begin{proof}

We may assume $|\psi\rangle \in \mathscr{P}^\natural_{\mathcal M}$
by invariance of the $L_p$-norms. In parts (a1), (a2) of this proof we 
first apply that $|\psi\rangle$ is faithful in order to apply the results by \cite{AM}.

(a1) Assume that $p_0,p_1 \in [1,2]$. This part of the proof is taken from paper I and only included for convenience.
Denote the dual of a H\" older index $p$ by $p'$, 
defined so that $1/p+1/p'=1$. \cite{AM} have shown that the non-commutative $L_p(\mathcal{M}, \psi)$-norm of a vector $|\zeta\rangle$ relative to $|\psi\rangle$
can be characterized by (dropping the superscript on the norm)
\ben
\| \zeta \|_{p,\psi} = \sup \{ |\langle \zeta | \zeta' \rangle | : \ \ \| \zeta' \|_{p',\psi} \le 1 \}. 
\een
They have furthermore shown (\cite{AM}, thm. 3) that when $p'\ge 2$, any vector $|\zeta'\rangle \in L_{p'}(\mathcal{M},\psi)$ has a unique generalized polar decomposition, i.e. can be written in 
the form $|\zeta'\rangle = u \Delta_{\phi,\psi}^{1/p'} |\psi\rangle$, where $u$ is a unitary or partial isometry from $\mathcal{M}$. Furthermore, they show that $\| \zeta' \|_{p',\psi} = \|\phi\|^{p'}$. We may thus choose a $u$ and a normalized $|\phi \rangle$, so that 
\ben
 \| G(\theta) \|_{p_\theta,\psi} = \langle u \Delta_{\phi,\psi}^{1/p_\theta'} \psi | G(\theta) \rangle, 
\een
perhaps up to a small error which we can let go zero in the end. Now we define $p_\theta$ as in the statement, so that 
\ben
\frac{1}{p_\theta'} = \frac{1-2\theta}{p_0'} + \frac{2\theta}{p_1'}, 
\een
and we define an auxiliary function $f(z)$ by 
\ben\label{eq:fdef}
f(z)=\langle u \Delta_{\phi,\psi}^{2\bar z/p_1'+(1-2\bar z)/p_0'} \psi | G(z) \rangle, 
\een
noting that 
\ben
f(\theta)= \| G(\theta) \|_{p_\theta,\psi} 
\een
by construction.
By Tomita-Takesaki-theory, $f(z)$ is holomorphic in $\bS_{1/2}$. For the values at the boundary of the strip $\bS_{1/2}$, we estimate
\ben\label{eq:est1}
\begin{split}
|f(it)| & = |\langle u \Delta_{\phi,\psi}^{-2it(1/p_1'-1/p_0')} \Delta_{\phi,\psi}^{1/p_0'} \psi | G(it) \rangle|  \\
        & \le \| u \Delta_{\phi,\psi}^{-2it(1/p_1'-1/p_0')} \Delta_{\phi,\psi}^{1/p_0'} \psi \|_{p_0',\psi} \| G(it) \|_{p_0,\psi}  \\
        & \le \| \Delta_{\phi,\psi}^{-2it(1/p_1'-1/p_0')} \Delta_{\phi,\psi}^{1/p_0'} \psi \|_{p_0',\psi} \| G(it) \|_{p_0,\psi}  \\
        & \le \|\phi\|^{p_0'} \| G(it) \|_{p_0,\psi}  \\
        & \le \| G(it) \|_{p_0,\psi}  .
\end{split}
\een
Here we used the version of H\" older's inequality proved by \cite{AM}, we used $ \| a^* \zeta \|_{p_0',\psi} \le \|a\| \| \zeta \|_{p_0',\psi}$ for 
any $a \in \cA$, see \cite{AM}, lem. 4.4, and we used $\| \Delta_{\phi,\psi}^{-2it(1/p_1'-1/p_0')} \Delta_{\phi,\psi}^{1/p_0'} \psi \|_{p_0',\psi} \le \|\phi\|^{p_0'}$ which we 
prove momentarily. A similar chain of inequalities also gives 
\ben
\label{eq:est2}
|f(1/2+it)| \le \| G(1/2+it) \|_{p_1,\psi}.
\een
To prove the remaining claim, let $|\zeta'\rangle = \Delta^{z}_{\phi,\psi}|\psi\rangle$ and $z=1/p'+2it$. Then we have, using the variational characterization 
by \cite{AM} of the $L_{p'}(\mathcal{M},\psi)$-norm when $p'\ge 2$:
\ben
\begin{split}
\|\zeta'\|_{p',\psi} =& \sup_{} \{ \| \Delta_{\chi,\psi}^{1/2-1/p'} \Delta^{z}_{\phi,\psi} \psi  \| : \| \chi \|=1 \}\\
=& \sup_{} \{ \| \Delta_{\chi,\psi}^{1/2-1/p'-2it} \Delta^{1/p'+2it}_{\phi,\psi}  \psi  \| : \| \chi \|=1 \}\\
=& \sup_{} \{ \| \Delta_{\chi,\psi}^{1/2-1/p'} (D\chi:D\phi)_{2t} \pi^{\mathcal M}(\phi) \Delta^{1/p'}_{\phi,\psi}  \psi \| : \| \chi \|=1 \} \\
\le & \sup_{} \{ \| \Delta_{\chi,\psi}^{1/2-1/p'} a \Delta^{1/p'}_{\phi,\psi}  \psi \| : \| \chi \|=1, a \in \cM, \|a\|=1 \} \\
\le & \sup_{} \{ \| a \Delta^{1/p'}_{\phi,\psi}  \psi  \|_{p',\psi} : a \in \cM, \|a\|=1 \}. 
\end{split}
\een
Using \cite{AM}, lem. 4.4, we continue this estimation as 
\ben
\le  
\sup_{a \in \cM, \|a\|=1} \| a \|   \|\Delta^{1/p'}_{\phi,\psi}  \psi  \|_{p',\psi}
= \|\phi\|^{p'}, 
\een
which gives the desired result. 

Next, we use the Hirschman improvement of the Hadamard three lines theorem \cite{Hirschman}.

\begin{lemma}\label{lem:4}
Let $g(z)$ be holomorphic on the strip $\bS_{1/2}$, continuous and uniformly bounded at the boundary of $\bS_{1/2}$. Then for $\theta \in (0,1/2)$, 
\ben
\ln |g(\theta)| \le \int_{-\infty}^\infty \left(
\beta_{\theta}(t) \ln |g(1/2+it)|^{2\theta} +  \alpha_{\theta}(t) \ln |g(it)|^{1-2\theta}
\right) \dd t,
\een
where $\alpha_\theta(t),\beta_\theta(t)$ are as in lem. \ref{lem:hirsch}.
\end{lemma}

Applying this to $g=f$ gives the statement of the theorem. 

\medskip

(a2) Now we assume that $p_0,p_1 \in [2,\infty]$. \cite{AM} have shown that 
for any\footnote{The cone $\sP_\cM^{1/(2p')}$ is defined as the closure of $\Delta^{1/(2p')}_\psi \cM_+ |\psi\rangle$ and its 
properties are discussed in \cite{AM}.} $\zeta'_+ \in L_{p'}^+(\cM,\psi):= L_{p'}$-closure of $\sP_\cM^{1/(2p')}$, $1 \le p' \le 2$ there is $\phi \in \sH$
such that for all $\zeta \in L_p(\cM,\psi)$ we have 
\ben
\langle \zeta'_+ | \zeta \rangle = \langle \Delta^{1/2}_{\phi,\psi} \psi | \Delta^{(1/p')-(1/2)}_{\phi,\psi} \zeta \rangle
\een
and such that $\|\zeta'_+\|_{p',\psi} = \|\phi\|^{2/p}$. Furthermore, by the non-commutative H\" older inequality 
proven in \cite{AM}, there exists $\zeta' \in L_{p_\theta'}(\cM,\psi)$ such that 
\ben
\| G(\theta) \|_{p_\theta,\psi} = \langle \zeta' | G(\theta) \rangle, \quad \|\zeta'\|_{p'_\theta,\psi} = 1.
\een
Thus, since by \cite{AM}, thm. 3 we may write 
$\zeta'=u\zeta'_+, u \in \cM$ with $u^*u \le 1$ and $\zeta'_+ \in L_{p'_\theta}^+(\cM,\psi)$, we have
\ben
\begin{split}
\| G(\theta) \|_{p_\theta,\psi} =& \langle \Delta^{1/2}_{\phi,\psi} \psi | \Delta^{1/p_\theta'-1/2}_{\phi,\psi} u^* G(\theta) \rangle\\
=& \langle \Delta^{1/2}_{\phi,\psi} \psi | \Delta^{(1-2\theta)/p_0' + (2\theta)/p_1' -1/2}_{\phi,\psi} u^* G(\theta) \rangle
\end{split}
\een
and $\|\phi\|=1$.
Similarly to the previous case we now consider the 
function
\ben\label{eq:fdef1}
f(z)= \langle \Delta^{1/2}_{\phi,\psi} \psi | \Delta^{(1-2z)/p_0' + (2z)/p_1' -1/2}_{\phi,\psi} u^* G(z) \rangle
, 
\een
which is holomorphic for $z \in \bS_{1/2}$ and uniformly bounded on the closure. For the lower boundary value we calculate
\ben
\begin{split}
|f(it)| =& |\langle \Delta^{1/2}_{\phi,\psi} \psi | \Delta^{-2it(1/p_0' -1/p_1')}_{\phi,\psi} \Delta_{\phi,\psi}^{1/p_0'-1/2} u^* G(it) \rangle | \\
\le & \|\Delta^{1/2}_{\phi,\psi} \psi\| \ \| \Delta_{\phi,\psi}^{1/p_0'-1/2} u^* G(it) \| \\
=& \|\phi\| \ \| \Delta_{\phi,\psi}^{1/2-1/p_0} u^* G(it) \| \\
\le& \sup\{ \| \Delta_{\chi,\psi}^{1/2-1/p_0} u^* G(it) \| : \| \chi \| = 1\} \\
=& \| u^* G(it) \|_{p_0,\psi} \le \| u^* \| \| G(it) \|_{p_0,\psi} = \| G(it) \|_{p_0,\psi}
\end{split}
\een
using in the last line the variational characterization of the $L_p$-norms and \cite{AM}, lem. 4.4. A similar chain of inequalities 
also gives $|f(1/2+it)| \le \| G(1/2+it) \|_{p_1,\psi}$. The rest follows from Hirschman's improvement as in the previous case. 
\medskip

(b) In the remaining part of the proof, we remove the faithful condition on the state $|\psi\rangle$.
Suppose that $\omega_\psi$ is non-faithful. 
For $\sigma$-finite $\cM$, there exists some cyclic and separating vector $|\eta\rangle$ for $\cM$ and we put 
\ben\label{rhorep}
\omega_{\psi_\epsilon} = (1-\epsilon) \ \omega_\psi + \epsilon \ \omega_\eta
\een
so that $|\psi_\epsilon\rangle \in \sP^\natural_\cM$ is now faithful for $\cM$ (and $\cM'$). 
The proof is then completed by the following lemma because we can apply part (a1),(a2) to the faithful state 
$|\psi_\epsilon\rangle$ and obtain b) by taking the limit $\epsilon \to 0$ and using the dominated convergence theorem
under the integral.

\begin{lemma}\label{monlp}
Let $\omega_\psi, \omega_\eta \in \sS(\cM)$, and let $\omega_{\psi_\epsilon} = (1-\epsilon) \ \omega_\psi + \epsilon \ \omega_\eta$. 
Then $\lim_{\epsilon \to 0+} \| \zeta \|_{p,\psi_\epsilon} =\| \zeta \|_{p,\psi}$ for any $p \ge 1$ and $|\zeta\rangle \in \sH$. 
\end{lemma}

\begin{proof}
(1) Case $p \ge 2$:
Clearly $\omega_{\psi_\epsilon} \ge (1-\epsilon) \omega_\psi$, from which it follows that $\Delta_{\phi,\psi_\epsilon} \le (1-\epsilon)^{-1} \Delta_{\phi,\psi}$
and thus by L\" owner's theorem \cite{Hansen}, $\Delta_{\phi,\psi_\epsilon}^{\alpha} \le (1-\epsilon)^{-\alpha} \Delta_{\phi,\psi}^{\alpha}$ for $\alpha \in [0,1]$, so by the variational definition of the $L_p$ norm (appendix):
\ben
\| \zeta \|_{p,\psi_\epsilon} \le (1-\epsilon)^{(1/p)-(1/2)} \| \zeta \|_{p,\psi} \quad \text{for $p\ge 2$.} 
\een
Therefore, by choosing $\epsilon > 0$ sufficiently small, we can achieve that 
\ben
\label{eq1}
\| \zeta \|_{p,\psi_\epsilon} - \| \zeta \|_{p,\psi} < \delta 
\een
for any given $\delta>0$. To get a similar inequality in the reverse direction, we use the following lemma.

\begin{lemma}\label{modconv}
Let $\omega_\psi, \omega_\eta, \omega_{\psi_n}, \omega_{\eta_n} \in \sS(\cM)$ be such that $\lim_n \| \omega_\psi -\omega_{\psi_n} \| = 0, 
\lim_n \| \omega_\eta -\omega_{\eta_n} \| = 0$ and such that $\omega_{\eta_n} \le C \omega_\eta, \omega_\psi \le C \omega_{\psi_n}$ for some $C<\infty$ and all $n$. Then
\ben
\lim_n \|(\Delta_{\eta,\psi}^{\alpha/2} -\Delta_{\eta_n,\psi_n}^{\alpha/2} ) \zeta \| = 0
\een
for any $\alpha \in [0,1), |\zeta\rangle \in \sD(\Delta_{\eta,\psi}^{\alpha/2})$.
\end{lemma}
\begin{proof}
We use the shorthands $\Delta = \Delta_{\eta,\psi}, \Delta_n = \Delta_{\eta_n,\psi_n}$. Without loss of generality $\alpha>0$.
To deal with the powers, we employ the standard formula
\ben
X^{\alpha} = \frac{\sin( \pi \alpha )}{\pi} \int_0^\infty d\lambda \, \lambda^{\alpha} \left[
\lambda^{-1} - (\lambda + X)^{-1}
\right]
\een
for $\alpha \in (0,1), X\ge 0$. We use this with $X=\Delta^{1/2}$ and $=\Delta_n^{1/2}$ giving us that
\ben
\begin{split}
&\| (\Delta^{\alpha/2} -\Delta_n^{\alpha/2} ) \zeta \| \\
\le & \  \int_{0}^\infty d\lambda \, \lambda^{\alpha-1} 
\left \|
\left[
(1 + \lambda\Delta^{-1/2})^{-1} - (1 + \lambda\Delta_n^{-1/2})^{-1}
\right] \zeta
\right \| . 
\end{split}
\label{prevlim}
\een
In the rest of the proof we denote by $c$ any constant depending only on $\alpha,C$.
We split the integration domain into three parts: $(0,\delta), (\delta, L), (L, \infty)$. 

(i) Range $(0,\delta)$: In this range, we use
\ben
\begin{split}
& \int_{0}^\delta d\lambda \, \lambda^{\alpha-1}
\left\| \left[
(1 + \lambda\Delta^{-1/2})^{-1} - (1 + \lambda\Delta_n^{-1/2})^{-1}
\right] \zeta
\right \| 
\\
= & \ \int_{0}^\delta d\lambda \, \lambda^{\alpha}
\left\| \left[
(\lambda + \Delta^{1/2})^{-1} - (\lambda + \Delta_n^{1/2})^{-1}
\right] \zeta
\right \| 
\\
\le & \ \int_{0}^\delta d\lambda \, \lambda^{\alpha} \left\{
\left\| 
(\lambda + \Delta^{1/2})^{-1} \zeta \right\| + \left\| (\lambda + \Delta_n^{1/2})^{-1}
 \zeta
\right \| \right\} 
\\
\le 2 \|\zeta\|
\int_{0}^\delta d\lambda \, \lambda^{\alpha-1} = c \| \zeta \| \delta^{\alpha}
\end{split}
\label{35}
\een
using that $\Delta, \Delta_n \ge 0$.  

(ii) Range $(\delta, L)$: By \cite{Araki2}, II, lem. 4.1, 
\ben
\left \|
\left[
(\lambda + \Delta^{1/2})^{-1} - (\lambda + \Delta_{n}^{1/2})^{-1}
\right] \zeta
\right \| \to 0 \quad \text{as $n \to \infty$, when $\lambda > 0$.}
\een
and the convergence is uniform for $\lambda$ in the 
compact set $[\delta, L]$.

(iii) Range $(L,\infty)$. The domination assumption gives $\Delta_n\le  C^2 \Delta$. The function $\RR_+ \owns x \mapsto (\lambda + x^{-1/2})^{-2}$
is operator monotone, thus by by L\" owner's theorem \cite{Hansen}:
\ben
\| (1 + \lambda\Delta_n^{-1/2})^{-1} \zeta\| = \langle \zeta | (1 + \lambda\Delta_n^{-1/2})^{-2} \zeta \rangle^{1/2}
\le  \langle \zeta | (1 + \lambda C^{-1} \Delta^{-1/2})^{-2} \zeta \rangle^{1/2}.
\een
and since $C \ge 1$ trivially 
\ben
\| (1 + \lambda\Delta^{-1/2})^{-1} \zeta\| = \langle \zeta | (1 + \lambda\Delta^{-1/2})^{-2} \zeta \rangle^{1/2}
\le  \langle \zeta | (1 + \lambda C^{-1} \Delta^{-1/2})^{-2} \zeta \rangle^{1/2}.
\een
Using these inequalities under the integral \eqref{prevlim} gives:
\ben
\begin{split}
& \int_{L}^\infty d\lambda \, \lambda^{\alpha-1} 
\left \|
\left[
(1 + \lambda\Delta^{-1/2})^{-1} - (1 + \lambda\Delta_n^{-1/2})^{-1}
\right] \zeta
\right \| \\
\le & \
\int_{L}^\infty d\lambda \, \lambda^{\alpha-1} \left\{
\left \|
(1 + \lambda\Delta^{-1/2})^{-1} \zeta
\right\| +
\left\| (1 + \lambda\Delta_n^{-1/2})^{-1}
 \zeta
\right \|
\right\} \\
\le & \
2 \int_{L}^\infty d\lambda \, \lambda^{\alpha-1} 
\langle \zeta | (1 + \lambda C^{-1} \Delta^{-1/2})^{-2} \zeta \rangle^{1/2} \\
\le & \
cL^{-\alpha/2} \left\{ \int_L^\infty 
d\lambda \, \lambda^{-1+\alpha} \ \langle \zeta | (1 + \lambda C^{-1} \Delta^{-1/2})^{-2} \zeta \rangle
\right\}^{1/2} \\
= & \
cL^{-\alpha/2}\bigg\{  \langle \zeta | f(C\Delta^{1/2}) \zeta\rangle \bigg\}^{1/2} \le cL^{-\alpha/2} \| \Delta^{\alpha/4} \zeta \|,
\end{split}
\een
uniformly in $n$.
Here we have applied Jensen's inequality to the probability measure $L^{\alpha} \lambda^{-1-\alpha} d\lambda$ on $(L,\infty)$ in the third step. 
We have also defined/estimated the non-negative function
\ben
f(x)=\int_L^\infty d\lambda \, \lambda^{-1+\alpha} (1+x^{-1}\lambda)^{-2} \le c x^\alpha. 
\een
Applying standard subharmonic analysis to the subharmonic function $z \mapsto \log \| \Delta^{\alpha z/2} \|$ in the strip $0 \le \Re z \le 1$, 
we have $ \| \Delta^{\alpha/4} \zeta \|^2 \le \| \zeta \| \| \Delta^{\alpha/2} \zeta\|$, giving  
\ben
\int_{L}^\infty d\lambda \, \lambda^{\alpha-1} 
\left \|
\left[
(1 + \lambda\Delta^{-1/2})^{-1} - (1 + \lambda\Delta_n^{-1/2})^{-1}
\right] \zeta
\right \| \le c( L^{-\alpha} \| \zeta \| \| \Delta^{\alpha/2} \zeta\| )^{1/2}.
\label{47}
\een
\medskip

Now we choose $\delta, L$ so small/large that the contributions from (i), (iii), i.e. \eqref{35}, \eqref{47} are $<\epsilon/3$ each (independently of $n$) and then $n$ so large that 
the contribution (ii) from $(\delta, L)$ is $<\epsilon/3$. Then the integral \eqref{prevlim} is $<\epsilon$ by (i), (ii), (iii), and the proof is complete.
\end{proof}

We can now complete the proof of lem. \ref{monlp}. 
We can pick a unit $|\phi\rangle$ such that $\| \zeta \|_{\psi,p} \le \| \Delta_{\phi,\psi}^{(1/2)-(1/p)} \zeta \| + \delta/2$ by 
the variational definition of the $L_p$ norm for $p \ge 2$. 
Lem. \ref{modconv} and the triangle inequality 
shows that there is an $\epsilon > 0$ such that  
\ben
\begin{split}
\| \zeta \|_{p,\psi} \le& \| \Delta_{\phi,\psi}^{(1/2)-(1/p)} \zeta \| + \delta/2 \\
\le& \| \Delta_{\phi,\psi_\epsilon}^{(1/2)-(1/p)} \zeta \|  + \| (\Delta_{\phi,\psi}^{(1/2)-(1/p)} -\Delta_{\phi,\psi_\epsilon}^{(1/2)-(1/p)} ) \zeta \| + \delta/2 \\
\le& \sup \{ \| \Delta_{\chi,\psi_\epsilon}^{(1/2)-(1/p)} \zeta \| : |\chi\rangle \in \sH, \| \chi \| = 1\} + \delta\\
=&  \| \zeta \|_{p,\psi_\epsilon} + \delta, 
\end{split}
\een
and this together with \eqref{eq1} gives $| \ \| \zeta \|_{p,\psi} - \| \zeta \|_{p,\psi_\epsilon} \ |<2\delta$. Since $\delta$ is arbitrarily small, 
the proof of lem. \ref{monlp} is complete when $p \ge 2$.

\medskip
(2) Case $1 \le p \le 2$: This proof has already appeared in paper I and is only included for convenience. 
Since by \eqref{rhorep} $\omega_{\psi_\epsilon} /(1-\epsilon) > \omega_\psi$, it now follows similarly as in part (1) of this proof that
 \begin{equation}
 \label{*}
 \left\| \zeta \right\|_{p,\psi} 
\leq  (1-\epsilon)^{(1/p)-(1/2)} \left\| \zeta \right\|_{p,\psi_\epsilon} \quad \text{for $1 \le p \le 2$}.
 \end{equation}
The $L_p$-norms $ \left\| \zeta \right\|_{p,\psi}^p$ may be considered for fixed $|\zeta\rangle$ as 
functionals of the state $\omega_\psi$, and as such they are convex. Indeed, let $D_s(\omega_\zeta' | \omega_\psi')$
be the sandwiched relative Renyi entropy relative between two functionals $\omega_\zeta',\omega_\psi'$ on $\cM'$
induced by vectors $|\zeta\rangle, |\psi\rangle$, related to the $L_p$-norms by $D_s(\omega_\zeta' | \omega_\psi')
=(s-1)^{-1} \log  \left\| \zeta \right\|_{2s,\psi}^{2s}$. The data processing inequality for this quantity (see e.g. \cite{Berta2}, thm. 14) in combination with standard arguments as in e.g. \cite{mueller}, proof of prop. 1 implies joint convexity in $\omega_\zeta',\omega_\psi'$.
This gives in combination with \eqref{rhorep} that (for $p=2s$)
\ben
\label{**}
\left\| \zeta \right\|_{p,\psi_\epsilon} \le (1-\epsilon) \left\| \zeta \right\|_{p,\psi} + \epsilon \left\| \zeta \right\|_{p,\eta}.
\een
Combining \eqref{*} with \eqref{**} implies the statement of lem. \ref{monlp} in the case $1 \le p \le 2$.
\end{proof}

\medskip
This completes the proof of lem. \ref{lem:hirsch}.
\end{proof}

\section{Multi-trace inequalities for von Neumann algebras}
\label{sec:3}

As applications of lem. \ref{lem:4} we now prove various inequalities that reduce to ''multi-trace inequalities'' in the case of 
finite type I factors. For simplicity, it will be assumed that $\omega_\psi$ is a faithful state on the von Neumann algebra $\cM$, meaning $\omega_\psi(m^*m)=0$ implies $m=0$ for all $m \in \cM$. 

\begin{corollary}\label{cor:1}
Let $a_1, \dots, a_n \in \cM_+$, $r \in (0,1], p \ge 2$. Then 
\ben
\frac{1}{r} \log \| a_1^r \cdots a^r_n \psi \|_{p/r,\psi}
\le \int_\RR dt \, \beta_{r/2}(t) \, \log \| a_1^{1+it} \cdots a^{1+it}_n \psi \|_{p,\psi} . 
\een
\end{corollary}
\begin{proof}
We choose $p_1=p, p_0=\infty, \theta = r/2$ and
\ben
G(z)=a_1^{2z} \cdots a^{2z}_n |\psi \rangle
\een
in lem. \ref{lem:4}. Then $\|G(z)\|$ is uniformly bounded on $\bS_{1/2}$
and $p_\theta = p/r$.  At the lower boundary of the strip:
\ben
\| G(it) \|_{p_0,\psi} = \| a_1^{2it} \cdots a^{2it}_n \psi\|_{\infty,\psi} = \| a_1^{2it} \cdots a^{2it}_n \| = 1
\een
because $a_k^{2it}$ are unitary operators (using the isomeric identification of $L_\infty(\cM, \psi) \owns a|\psi\rangle \mapsto a \in \cM$
proven in \cite{AM}.) Thus the term from the lower boundary does not contribute and we obtain the statement.
\end{proof}

Another corollary of a similar nature is:

\begin{corollary}\label{cor:2}
(Araki-Lieb-Thirring inequality)
For $r \ge 2, |\psi\rangle,|\zeta\rangle \in \sH$ there holds 
\ben
\| \zeta \|_{r,\psi}^2 \le 
\|\Delta^{r/4}_{\zeta,\psi} \psi \|^{4/r}_{}.
\een
\end{corollary}
\begin{proof}
A proof for this has already been given in \cite{Berta2}, thm. 12, so the only point is to show an alternative proof.  
We may assume that $\|\Delta^{r/4}_{\zeta,\psi} \psi \| < \infty$, otherwise the statement is trivial. Also, we may assume 
without loss of generality that $|\zeta\rangle$ is in the natural cone. In lem. \ref{lem:4}, we take $G(z) = \Delta^{rz/2}_{\zeta,\psi} \psi$, 
$p_1 = 2, p_0 = \infty, \theta = 1/r$, so $p_\theta = r$. Then $G(z)$ is holomorphic and uniformly bounded in $\bS_{1/2}$, see e.g. 
lem. 3 of \cite{Araki4}. 

On the left side of lem. \ref{lem:4}  we 
obtain $\log \| \Delta_{\zeta,\psi}^{1/2} \psi \|_{r,\psi}^r = \log \| \zeta \|_{r,\psi}^r$.
We compute at the lower boundary of the strip:
\ben
\| G(it) \|_{p_0,\psi} = \| \Delta^{irt/2}_{\zeta,\psi}  \psi\|_{\infty,\psi} = \| \Delta^{irt/2}_{\zeta,\psi} \Delta^{-irt/2}_{\psi,\psi} \psi\|_{\infty,\psi} = 
\| u(rt/2) \psi\|_{\infty,\psi} = \| u(rt/2)\| = 1. 
\een
Here $u(t) = \Delta^{it}_{\zeta,\psi} \Delta^{-it}_{\psi,\psi}$ is the Connes cocycle which is a unitary from $\cM$ 
and we used again the isomeric identification of $L_\infty(\cM, \psi) \owns a|\psi\rangle \mapsto a \in \cM$
proven in \cite{AM}. Thus the term from the lower boundary does not contribute. At the upper boundary of the strip:
\ben
\| G(1/2 +it) \|_{p_1,\psi} = \| \Delta^{irt/2+r/4}_{\zeta,\psi}  \psi\|_{2,\psi} = \| \Delta^{r/4}_{\zeta,\psi}  \psi\|,
\een
which no longer depends upon $t$, using that the $L_2$ norm is equal to the Hilbert space norm \cite{AM} and 
that $\Delta^{it}_{\zeta,\psi}$ is a unitary operator. 
Since $\int dt \beta_\theta(t) = 1$ we obtain the statement.
\end{proof}

Let $h$ be a self-adjoint element of $\cM$ and $|\psi\rangle \in \sH$ a normalized state vector. Following Araki \cite{Araki5}, the non-normalized 
perturbed state $|\psi^h\rangle$ is defined by the absolutely convergent series
\ben \label{psih}
|\psi^h\rangle = \sum_{n=0}^\infty \int_0^{1/2} ds_1 \dots \int_0^{s_{n-1}} ds_n \, \Delta_\psi^{s_n}h\Delta^{s_{n-1}-s_n}_\psi h 
\dots \Delta^{s_{1}-s_2}_\psi h |\psi\rangle,
\een
which can also be written as $e^{(\log \Delta_\psi+h)/2}|\psi\rangle$ \cite{Araki4}. 
This technique of perturbations has been generalized to semi-bounded -- instead of bounded -- operators 
by \cite{Donald1}, see also \cite{Petz1993}, sec. 12. The perturbations, $h$ that would normally be in $\cM_{\rm s.a.}$ 
are in this framework generalized to so-called ``extended-valued upper bounded 
self-adjoint operators affiliated with $\cM$'', the space of which is called $\cM_{\rm ext}$. 
More precisely, $h \in \cM_{\rm ext}$ if 
\begin{enumerate}
\item[(i)] it is a linear, upper semi-continuous map $\sS(\cM) \owns \sigma \mapsto \sigma(h) \in \RR \cup \{\infty\}$, and
\item[(ii)] the set $\{ \sigma(h) : \sigma \in \sS(\cM) \}$ is bounded from above.
\end{enumerate}
For any ``operator'' $h \in \cM_{\rm ext}$, one shows that it is consistent to define:
\begin{definition}
(see \cite{Donald1}, thm. 3.1) 
If $h \in \cM_{\rm ext}$,
the perturbed state $\sigma^{h}$ of a normal state $\sigma \in \sS(\cM)$ , is given by the unique extremizer of the convex variational 
problem
\ben
c(\sigma,h) = \sup \{ \rho(h) - S(\rho | \sigma) : \rho \in \sS(\cM) \}
\een 
provided the sup is not $-\infty$. 
\end{definition}
The latter is certainly the case if $h \in \cM_{\rm s.a.}$ is an ordinary self-adjoint 
element of the von Neumann algebra $\cM$, and in this case the above ``thermodynamic'' definition of the perturbed 
state is up to normalizations equivalent to Araki's ``perturbative'' definition \eqref{psih}:
\ben
c(\sigma,h) = \log \| \eta^h \|^2, \quad \sigma^h(m) = \langle \eta^h | m  | \eta^h\rangle/\| \eta^h\|^2,
\een
wherein $|\eta\rangle$ is a vector representer of the state $\sigma$, see \cite{Donald1}, ex. 3.3. 
Furthermore, $h \in \cM_{\rm ext}$ has the spectral decomposition \cite{Donald1}, prop. 2.13 (B) 
\ben
\label{specrep}
h = \int_{-\infty}^{c} \lambda E_{h}(d\lambda) - \infty \cdot q . 
\een
Here, $q \in \cM$ is a projector onto the subspace where $h$ is $-\infty$, 
and the measure $E_h(d\lambda)$ takes values in the projections in 
$(1-q)\cM(1-q)$, so it commutes with $q$.

\begin{corollary}
\label{cor:3}
(Generalized Golden-Thomson inequality)
For $h_i \in \cM_{\rm ext}, |\psi \rangle \in \sH, \|\psi\|=1$ there holds 
\ben
\log \| \psi^{h_1+\dots+h_k} \|^2 \le  \int_\RR dt \, \beta_{0}(t) \, \log \left\{ \| \prod_{j=1}^k e^{(1/2+it)h_j}  \psi \| \ \| \prod_{j=k}^1 e^{(1/2-it)h_j}  \psi \| \right\}.
\een
\end{corollary}

\begin{proof}
Case I). First we assume each $h_j \in \cM_{\rm s.a.}$, i.e. it is bounded.
We let 
\ben
G(z) = \Delta_\psi^{z/2} e^{zh_1} \dots e^{zh_k} |\psi\rangle. 
\een
By standard results of Tomita-Takesaki theory, 
this family of vectors is analytic on $\bS_{1/2}$ and uniformly bounded in the norm of $\sH$ on the closure, for instance by the maximum of $1$ and
$\prod_{i=1}^k \| e^{h_i} \|$ using a standard Phragmen-Lindel\" of type argument. In lem. \ref{lem:4}, we use this with
$p_1 = 2, p_0 = \infty, \theta = 1/n, n \in 2\mathbb N$, so $p_\theta = n$. 
At the lower boundary of $\bS_{1/2}$, we get $\|G(it)\|_{2,\psi} = 1$ -- the $L_2$-norm is the Hilbert space norm -- 
so this does not contribute. 
Keeping therefore only the term from the upper boundary, we have
\ben\label{eq:pr1}
\log \|\Delta_\psi^{1/(2n)} e^{h_1/n} \cdots e^{h_k/n} \psi\|_{\psi,n}^n \le 
\int_\RR dt \, \beta_{1/n}(t) \, \log \| \Delta_\psi^{1/4} e^{(1/2+it)h_1} \cdots e^{(1/2+it)h_k} \psi \|^2.
\een 
Now we consider the left side, putting $a_n = e^{h_1/n} \cdots e^{h_k/n}$. By \cite{AM}, thm. 3 (4), there exists\footnote{The cone 
$\sP_{\cM}^{1/(2n)}$ is defined as the closure of $\Delta_\psi^{1/(2n)} \cM_+ |\psi\rangle$ in $\sH$.} 
$|\phi_n\rangle \in L_n(\sH,\psi) \cap {\mathscr P}^{1/(2n)}_\cM$ such that 
\ben\label{eq:pr2}
\Delta_{\phi_n,\psi}^{1/n} |\psi\rangle = \Delta_{\psi}^{1/(2n)} a_n |\psi\rangle, \quad \| \phi_n \|^2 = 
\|\Delta_\psi^{1/(2n)} a_n \psi\|_{\psi,n}^n.
\een
It follows that
\ben\label{eq:pr3}
|\phi_n\rangle = J\Delta_{\phi_n,\psi}^{1/2} |\psi\rangle = J(\Delta_\psi^{1/(2n)} a_n \Delta_\psi^{1/(2n)})^{n/2}_{ }|\psi\rangle
\een
by a straightforward repeated application of \cite{AM}, lem. 7.7 (2); for the details see e.g. \cite{jaekel}, lem. 4.1.
Combining \eqref{eq:pr1}, \eqref{eq:pr2}, \eqref{eq:pr3}, we arrive at 
\ben\label{eq:pr4}
\log \|(\Delta_\psi^{1/(2n)} e^{h_1/n} \cdots e^{h_k/n} \Delta_\psi^{1/(2n)})^{n/2}_{} \psi\|^2 \le 
\int_\RR dt \, \beta_{1/n}(t) \, \log \| \Delta_\psi^{1/4} e^{(1/2+it)h_1} \cdots e^{(1/2+it)h_k} \psi \|^2.
\een 
We now take the limit $n \to \infty$ on the left side. Araki's version of the Lie-Trotter formula (suitably generalized to 
$k$ operators $h_1, \dots, h_k$, using that $e^{h_1/n} \cdots e^{h_k/n} = 1+n^{-1}(h_1+\dots+h_k) + O(n^{-2})$ where 
$\|O(n^{-2})\| \le Cn^{-2}$ for all $n>0$) see \cite{Araki4}, rem.s 1 and 2, establishes that 
\ben\label{LT}
s-\lim_n (\Delta_\psi^{1/(2n)} e^{h_1/n} \cdots e^{h_k/n} \Delta_\psi^{1/(2n)})^{n/4}_{} |\psi\rangle = |\psi^{h_1+\dots+h_k}\rangle
= e^{(\log \Delta_\psi + h_1+\dots+h_k)/2} |\psi\rangle, 
\een
so we get
\ben\label{eq:pr5}
\log \|\psi^{h_1+\dots+h_k} \|^2 \le 
\int_\RR dt \, \beta_{0}(t) \, \log \| \Delta_\psi^{1/4} e^{(1/2+it)h_1} \cdots e^{(1/2+it)h_k} \psi \|^2.
\een
On the integrand we finally use the following well-known application of the Hadamard three lines theorem ($0 \le \alpha < 1/2, m \in \cM$), 
\ben
\|\Delta_\psi^\alpha m \psi \| \le \| \Delta_\psi^{1/2} m \psi \|^{2\alpha} \| m \psi \|^{1-2\alpha} =  \| m^* \psi \|^{2\alpha} \| m \psi \|^{1-2\alpha}
\een
using that $z\mapsto \log \| \Delta_\psi^z m \psi \|$ is subharmonic on $\bS_{1/2}$. Using this with $\alpha = 1/4, m=e^{(1/2+it)h_1} \cdots e^{(1/2+it)h_k}$ gives the statement of the 
corollary.

Case II). The proof can be generalized to the case when $h_j \in \cM_{\rm ext}$ by reducing to the case I) via an approximation argument: 
Elements $k \in \cM_{\rm ext}$ can be approximated 
by bounded self-adjoint elements $k_n \in \cM_{\rm s.a.}$ by introducing a cutoff 
in the spectral decomposition \eqref{specrep}, as in  
\ben
\label{specrep}
k_n = \int_{-n}^{c} \lambda E_{k}(d\lambda) - n \cdot q \quad ;
\een
in fact one shows that $|\psi^{k_n}\rangle \to |\psi^k\rangle$ strongly, 
see \cite{Donald1}, prop. 3.15. We perform this cutoff for every $h_j$ obtaining a $h_{j,n}$. 
Since the desired inequality holds for $h_{j,n}$ by case I), the proof is completed by the fact that 
$e^{(1/2+it)h_{j,n}} \to e^{(1/2+it)h_{j}}$ as $n \to \infty$ strongly and uniformly in $t$ (as can be seen 
by decomposing $\sH = q_j\sH + (1-q_j)\sH$).
\end{proof}

\noindent
{\bf Examples:} 1) In the previous corollary we take $k=1, h_1=h$. Then the norm in the integrand no longer depends upon $t$ and we can use that $\int dt \beta_0(t) = 1$ to get:
\ben\label{GT}
\|\psi^h\| \le \| e^{h/2} \psi\|, 
\een
as shown previously by \cite{Araki4}. 

\medskip
\noindent
2) Finite-dimensional type I algebras. Let $\mathcal{A} = M_{n}({\mathbb C})$.
We will work in the standard Hilbert space ($\sH \simeq M_{n}({\mathbb C}) \simeq
{\mathbb C}^{n *} \otimes  \mathbb{C}^n$) and identify state functionals such as $\omega_\psi$ with density matrices
via $\omega_\psi(a) = \tr(a\omega_\psi)$. Vectors $|\zeta\rangle$ in $\sH$ are thus identified with matrices $\zeta \in M_n(\CC)$.
We have already mentioned that the $L_p(\cA,\psi)$-norms can be computed using the well known correspondence between these norms and the sandwiched relative entropy discussed in \cite{Berta2}:
$
\| \zeta \|_{p,\psi}^p = \tr(\zeta \rho_{\psi}^{2/p-1} \zeta^*)^{p/2}
$
where $|\zeta\rangle \in \sH$ is identified with a matrix $\zeta \in M_n(\CC)$ as described. Let $a_i$ be non-negative matrices. The multi-matrix inequality in cor. \ref{cor:1} then reads, 
when $\omega_\psi$ is the normalized tracial state $\omega_\psi(a)= \tr(a)/n$,
\ben
 \log \tr |a_1^{r} \cdots a_k^r|^{p/r} \le  \int_\RR dt \, \beta_{r/2}(t) \, \log \tr | a_1 a_2^{1+it} \cdots a^{1+it}_{k-1} a_k |^p ,
\een
which generalzes the Araki-Lieb-Thirring inequality (corresponding to $k=2$). This has been derived previously in 
\cite{wilde4,sutter2}, so our result can be seen as a generalization of these results to arbitrary von Neumann algebras. 
Cor. \ref{cor:2} is another generalization of this inequality which gives nothing new in the present case. 
Cor. \ref{cor:3} gives the following inequality. Under the above identification of vectors $|\psi\rangle \in \sH$ and matrices, the perturbed vector is
\ben
|\psi^h\rangle = |e^{\log \psi + h/2}\rangle
\een
(assuming $|\psi\rangle$ to be in the natural cone, i.e. self-adjoint and non-negative), and then choosing $|\psi = 1_n/\sqrt{n}\rangle$ as the vector representing the tracial state on $\cA$, 
we have
\ben
\log \tr e^{h_1+\dots+h_k}  \le  \int_\RR dt \, \beta_{0}(t) \, \log \tr | e^{(1/2)h_1} e^{(1/2+it)h_2} \cdots e^{(1/2+it)h_{k-1}} e^{(1/2)h_k}|^2 ,
\een
for any hermitian matrices $h_i$.
This reduces to the Golden-Thomson inequality for $k=2$,
\ben
\tr e^{h_1+h_2} \le \tr (e^{h_1} e^{h_2}),
\een
using that the trace in the integrand no longer depends on $t$ and $\int dt \beta_0(t) = 1$. For arbitrary number of matrices 
this is due to \cite{Sutter2}, who also explain the relation with Lieb's triple matrix inequality (for $k=3$).

\section{Improved DPI and recovery channels}
\label{sec:4}
\subsection{Relative entropy and measured relative entropy}
\label{subsec:2}

For the von Neumann algebra $\cA = M_n(\CC)$, the relative entropy between two states (density matrices) $\omega_\psi, \omega_\eta$ 
is defined by:
\ben
\label{eq:Srel}
S(\omega_\psi | \omega_\eta) = \tr(\omega_\psi \log \omega_\psi - \omega_\psi \log \omega_\eta).
\een
This may be expressed in terms of the logarithm of the relative modular operator in \eqref{modobj}, and this observation is the basis 
for Araki's approach \cite{Araki1,Araki2} to relative entropy for general von Neumann algebras. The main technical difference 
in the general case is that the individual terms in the above expression such as the von Neumann entropy 
$-\tr(\omega_\psi \log \omega_\psi)$ are usually infinite. Thus form a mathematical viewpoint, the relative- and 
not the absolute entropy is the primary concept.

Let $(\cM, J, \sP_\cM^\natural, \sH)$ be a von Neumann algebra in standard form acting on a Hilbert space $\sH$, 
with natural cone $\sP^\sharp_\cM$ and modular conjugation $J$. 
According to \cite{Araki1,Araki2}, if $\pi^\cM(\eta) \ge \pi^\cM(\psi)$, the relative entropy may be defined in terms of them by\footnote{The limit 
exists under this condition but may be equal to $+\infty$.} 
\ben
\label{eq:Sdef}
S(\psi | \eta) = -\lim_{\alpha \to 0^+} \frac{\langle \xi_\psi | \Delta^\alpha_{\eta, \psi} \xi_\psi \rangle-1}{\alpha} ,  
\een
otherwise, it is by definition infinite. Here, $|\xi_\psi\rangle$ denotes the unique representer of 
a vector $|\psi\rangle$ in the natural cone. The relative entropy only depends on the 
functionals $\omega_\psi, \omega_\eta$ on $\cM$, but not the choice of vectors $|\psi\rangle, |\eta\rangle$ that define these functionals. 
We will therefore use interchangeably the notations $S(\psi | \eta)=S(\omega_\psi | \omega_\eta)$. 
Araki's definition of $S(\omega_\psi | \omega_\eta)$ still satisfies the data processing inequality \eqref{dpi} \cite{Uhlmann1977} along with 
many other properties, see e.g. \cite{Petz1993}.

For $t \in \RR$,  the Connes-cocycle $(D\psi : D\eta)_{t}$ is the isometric operator from $\cM$ satisfying 
\ben
(D\psi : D\eta)_{t}  \pi^{\cM'}(\psi) = \Delta^{it}_{\psi,\psi} \Delta^{-it}_{\eta,\psi}.
\een
It only depends on the state functionals $\omega_\psi, \omega_\eta$.
In terms of the Connes-cocycle, the relative entropy \eqref{eq:Sdef} may also be defined as
\ben
S(\omega_\psi | \omega_\eta) \equiv S(\psi | \eta) = -i\frac{\dd}{\dd t} \omega_\psi( (D\eta : D\psi)_{t} ) |_{t=0}.  
\een
The last expression has the advantage that it does not require one to know the vector representative 
of $|\psi \rangle$ in the natural cone; in particular it shows that $S$ only depends on the state functionals.\footnote{The derivative exists whenever $S(\psi | \eta)<\infty$ \cite{Petz1993}, thm. 5.7.}

Later we will use the following variational expression for the relative entropy \cite{hiai}, prop. 1,
\ben
\label{eq:var}
S(\psi | \eta) = \sup_{h \in \cM_{\rm s.a.}} \{ \omega_\psi(h) - \log \| \eta^h \|^2 \}, 
\een
with $\cM_{\rm s.a.}$ the set of self-adjoint elements of $\cM$.
A related variational quantity is the {\bf ``measured relative entropy''}, $S_{\rm meas}$, defined as
\ben
\label{SMdef}
S_{\rm meas}(\psi | \eta) = \sup_{h \in \cM_{\rm s.a.}} \{ \omega_\psi(h) - \log \| e^{h/2} \eta \|^2 \}. 
\een
From the Golden-Thomson inequality \eqref{GT} we find
\ben
S_{\rm meas}(\psi | \eta) \le S(\psi | \eta).
\een
$S_{\rm meas}$ can also be written in terms of the classical relative entropy $S(\mu|\nu)$ (Kullback-Leibler divergence) of two probability measures
\ben
S(\mu | \nu) = \int d\mu \ \log \frac{d\mu}{d \nu} 
\een
as follows. Let $a \in\cM_{\rm s.a.}$ be a self-adjoint element of $\cM$. Then it has a spectral decomposition 
\ben
a = \int \lambda E_a(d\lambda)
\een
with an $\cM$-valued projection measure $E_a(d\lambda)$. Given $|\psi\rangle, |\eta \rangle \in \sH$, 
we get Borel measures $d\mu_{\psi,a} = \langle \psi | E_{a}(d\lambda) \psi \rangle$, and likewise for $|\eta\rangle$. 
Physically, these correspond to the probability distributions for measument outcomes of $a$ in the states $|\psi\rangle$ resp. $|\eta\rangle$. 
The relative entropy between these measures is defined (but can be $+\infty$) if ${\rm supp} \mu_{\eta,a} \subset {\rm supp} \mu_{\psi,a}$, wherein 
$d\mu_{\psi,a}/d\mu_{\eta,a}$ means the Radon-Nikodym derivative between the measures. We may perform the 
maximization in \label{SMdef} over $f(h)$ with\footnote{
More precisely, the space $L^\infty$ is defined relative to the measure $\mu_{h,\psi}$ relative to some 
faithful normal state $\psi \in \sS(\cM)$. Depending on the nature of this measure, ``$L^\infty$'' 
means either $\ell^\infty(\{1,\dots,n\}), \ell^\infty({\mathbb N})$ or $L^\infty(\RR)$ or a combination thereof, 
wherein the counting measure is understood in the first two cases, whereas the Lebesgue measure is understood 
in the last case.} $f \in L^\infty(\RR; \RR)$ and $h \in \cM_{\rm s.a.}$ because $f(h) \in \cM_{\rm s.a.}$. 
Maximizing first for fixed $h$ 
over $f$ and using ($=$ eq. \eqref{eq:Srel} in the commutative case)
\ben
\sup\left\{ \int f d\mu - \log \int e^f d\nu: f \in L^\infty(\RR;\RR) \right\} = S(\mu | \nu), 
\een
we can write the measured relative entropy in the following way:
\ben
\begin{split}
S_{\rm meas}(\omega_\psi | \omega_\eta) =& \sup\{ S(\mu_{h,\psi} | \mu_{h,\eta}) : h \in \cM_{\rm s.a.}\}\\
=& \sup \{ S(\omega_{\psi | \cC} | \omega_{\eta | \cC}) : \cC \subset \cM \ \ \text{a commutative von Neumann subalgebra} \}.
\end{split}
\een
This motivates the name ``measured relative entropy''. The second equality holds by
\cite{Petz1993}, prop. 7.13, for a related discussion see also  \cite{Berta}, lem. 1 
which corresponds to counting measures on the finite set $\{1,\dots,n\}$. 

For later we would like to know the relationship between $S_{\rm meas}$ and the fidelity, $F$. According to \cite{UhlmannFidelity}, the 
fidelity between two states $\omega_\eta, \omega_\psi \in \sS(\cM)$ on a von Neumann algebra $\cM$ in standard form may be defined as 
\ben
\label{supdef}
F(\omega_\psi | \omega_\eta) = \sup \{ |\langle \eta |u' \psi \rangle| : u' \in \cM', \|u'\|=1\}.
\een
It is related to the $L_1$-norm relative to $\cM'$ by $F(\omega_\psi | \omega_\eta) = \| \eta \|_{1,\psi,\cM'}$, see e.g. paper I, lem. 3 (1).
We claim:
\begin{proposition}\label{FS}
If $\omega_\eta \in \sS(\cM)$ is a faithful state on the von Neumann algebra $\cM$, 
then $S_{\rm meas}(\omega_\psi | \omega_\eta) \ge -\log F(\omega_\psi | \omega_\eta)^2$.
\end{proposition}

\begin{proof}
We may assume at  that $|\eta\rangle$ is cyclic for $\cM$, for 
if not we can obtain an equivalent standard form of $\cM$ after a GNS-construction based on $\omega_\eta$ and work with that standard form. 
Without loss of generality, $|\eta\rangle \in \sP_{\cM}^{\sharp}$.
Consider in $L_1(\cM',\eta)$ the polar decomposition $|\psi\rangle = u'^*|\psi_+\rangle$ into a $u' \in \cM'$ such that 
$u'^{*} u'= \pi^{\cM'}(\psi) \le 1$
and $|\psi_+\rangle \in \sP_{\cM'}^{1/2}$, see \cite{AM}, thm. 3. By definition, the cone $\sP_{\cM'}^{1/2}$ is the 
closure of $\Delta_\psi^{\prime 1/2} \cM'_+|\eta\rangle$, which equals the closure of $\cM_+|\eta\rangle$, since 
$J\Delta_\psi^{\prime 1/2} a'|\eta\rangle = a'|\eta\rangle$ for $a' \in \cM'_+$, $J|\eta\rangle = |\eta\rangle$ and $J\cM'J=\cM$. 
Thus, there exists a sequence $\{a_n\} \subset \cM_+$ such that $\lim_n a_n|\eta\rangle = u'|\psi\rangle$ strongly, so 
\ben
\label{limit}
\lim_n \langle \eta |a_n \eta \rangle = \langle \eta |u' \psi \rangle \in \RR_+.
\een
Then, with $E_{a_n}(d\lambda)$ 
the spectral decomposition of $a_n$ and $d\mu_{a_n,\psi} = \langle \psi | E_{a_n}(d\lambda) \psi \rangle$, 
$d\mu_{a_n,\eta} = \langle \eta | E_{a_n}(d\lambda) \eta \rangle$, the definition of the measured relative entropy and Jensen's inequality
applied to the convex function $-\log$ yields
\ben\label{85}
S_{\rm meas}(\omega_\psi | \omega_\eta) \ge S(\mu_{a_n,\psi} | \mu_{a_n,\eta}) \ge -2 \log \int \left(
\frac{d\mu_{a_n,\psi}}{d\mu_{a_n,\eta}}
\right)^{1/2} d\mu_{a_n,\eta} = -2 \log F(\mu_{a_n,\psi} | \mu_{a_n,\eta} ),
\een
where the Radon-Nikodym derivative 
is defined since $|\eta\rangle$ is faithful. The strong limit $\lim_n a_n|\eta\rangle = u'|\psi\rangle$ and 
$d\mu_{a_n,\psi} = \langle u'\psi | E_{a_n}(d\lambda) u'\psi \rangle$ (because $u' \in \cM', u'^{*} u' = \pi^{\cM'}(\psi)$ and $E_{a_n}$ takes 
values in $\cM$) imply that $\|\mu_{a_n,\psi} - \mu_{a_n,a_n\eta}\|_1 \le \| \omega_\psi - \omega_{a_n\eta}\| \le \| \psi+a_n\eta\| \ \| \psi - a_n \eta \| \to 0$ as $n \to \infty$. By paper I, lem. 11 and \eqref{CS} applied to the commutative case, this gives that also 
\ben
| F(\mu_{a_n,\psi} | \mu_{a_n,\eta} ) - F(\mu_{a_n,a_n\eta} | \mu_{a_n,\eta} ) | \le \|\mu_{a_n,\psi} - \mu_{a_n,a_n\eta}\|_1^{1/2}  \to 0.
\een 
By definition,
\ben
\left( \frac{d\mu_{a_n,a_n \eta}(\lambda)}{d\mu_{a_n,\eta}(\lambda)}
\right)^{1/2} = \lambda  \quad \text{for $\lambda \in \RR_+$,}
\een
hence by \eqref{85}
\ben
\begin{split}
&S_{\rm meas}(\omega_\psi | \omega_\eta) \ge -2 \log \lim_n  \int \lambda d\mu_{a_n,\eta} = 
-2 \log \lim_n  \int \lambda \langle \eta| E_{a_n}(d\lambda) \eta \rangle \\
=& 
-2 \log \lim_n  \langle \eta| a_n \eta \rangle = -2 \log  \langle \eta| u' \psi \rangle = -2 \log | \langle \eta| u' \psi \rangle |.
\end{split}
\een
The right side is by definition $\ge -\log F(\omega_\psi | \omega_\eta)^2$ as $\|u'\| = 1, u' \in \cM'$, which concludes the proof. 
\end{proof}

\subsection{Petz recovery map}
\label{subsec:petz}

We now recall the definition of the Petz map in the case of general von Neumann algebras, discussed in more detail in \cite{Petz1993}, sec. 8.
Let $T: \cB \to \cA$ be a $^*$-preserving linear map between two von Neumann algebras $\cA, \cB$ in standard form acting on Hilbert 
spaces $\sH, \sK$. If 
\ben\label{eq:eqschwarz}
\begin{pmatrix}
\langle \zeta_1 | & \langle \zeta_1 | 
\end{pmatrix}
T \left(
\left[
\begin{matrix}
a  & b  \\
c  & d  
\end{matrix}
\right] 
\left[
\begin{matrix}
a^* & c^* \\
b^* & d^*
\end{matrix}
\right]
\right) 
\left(
\begin{matrix}
|\zeta_1\rangle \\
|\zeta_2 \rangle
\end{matrix}
\right)
\ge 0, \quad \forall |\zeta_i\rangle \in \sH, \quad T(1_\cB) = 1_\cA,
\een
and for all $a,b,c,d \in \cB$,
then $T$ is called 2-positive and unital. 
In the matrix inequality, we mean $T$ applied to each matrix element. By duality between $\cA$ and $\sS(\cA)$, 
$T:\cB \to \cA$ gives a corresponding map $\tilde T:\sS(\cA) \to \sS(\cB)$ by $\omega \mapsto \tilde T(\omega) := \omega \circ T$. 
For finite dimensional von Neumann algebras $\cA, \cB$ where state functionals are identified with density matrices 
through $\omega(a) = \tr(\omega a)$, we can think of $\tilde T$ as the linear operator on density matrices defined by
\ben
\tr \omega T(b) = \tr \tilde T(\omega) b \quad \forall b \in \cB.
\een
This operator $\tilde T$ is completely positive and trace-preserving. The quantum data processing inequality (DPI) \cite{Uhlmann1977} states that 
\ben
S(\omega_\psi | \omega_\eta) \ge S(\omega_\psi \circ T | \omega_\eta \circ T), 
\een
where the right side could also be written as $S(\tilde T(\omega_\psi) | \tilde T(\omega_\eta))$.

We recall the definition of the Petz-map. Let $|\eta_\cA\rangle$ be a cyclic and separating vector in the natural cone of a von Neumann algebra $\cA$ in standard form. Then the KMS scalar product on $\cA$ is defined as 
\ben
\langle a_1,a_2 \rangle_\eta = \langle \eta_\cA | a_1^* \Delta_\eta^{1/2} a_2 \eta_\cA \rangle . 
\een
 Let $\omega_\eta$ be the normal state functional on $\cA$ associated with $|\eta_\cA\rangle$. Then its pull-back $ \omega_\eta \circ T$ to $\cB$, which is also faithful\footnote{This follows from Kadison's inequality $T(b^*b) \ge T(b)^* T(b)$.} has a vector representative $|\eta_\cB\rangle \in \sK$ in the natural cone. So: 
\ben
\omega_\eta(a) = \langle \eta_\cA | a \eta_\cA \rangle, \quad 
\omega_\eta \circ T(b) = \langle \eta_\cB | b \eta_\cB \rangle.
\een
$|\eta_\cA\rangle$ resp. $|\eta_\cB\rangle$ give KMS scalar products for $\cA$ resp. $\cB$, which we can use to
define the adjoint $T^+:\cA \to \cB$ (depending on the choices of these vectors) of the normal, unital and 2-positive 
$T: \cB \to \cA$, which is again normal, unital, and 2-positive, see \cite{Petz1993} prop. 8.3. 
For finite dimensional matrix algebras $T^+$ corresponds dually to the linear operator $\tilde T^+$ acting on density matrices $\rho$ for 
$\cB$ given by
\ben
\tilde T^+(\rho) = \sigma_\cA^{1/2} T\left( \sigma^{-1/2}_\cB \rho \sigma^{-1/2}_\cB \right) \sigma_\cA^{1/2},
\een
wherein $\sigma_\cA$ is the density matrix of $|\eta_\cA\rangle$ and $\sigma_\cB = \tilde T(\sigma_\cA)$ for $|\eta_\cB\rangle$.
 The rotated Petz map, which we call $\alpha_{\eta,T}^t: \cA \to \cB$, 
is defined by conjugating this with the respective modular flows, i.e. 
\ben
\alpha_{\eta,T}^t = \varsigma_{\eta,\cB}^t \circ T^+ \circ \varsigma_{\eta,\cA}^{-t}
\een
where $\varsigma_{\eta,\cA}^t = {\rm Ad} \Delta_{\eta, \cA}^{it}$ is the modular flow for $\cA, |\eta_\cA\rangle$ etc. 
For finite dimensional matrix algebras, $\alpha_{\eta, T}^t$ gives by duality a linear operator $\tilde \alpha_{\eta,T}^t$ acting on density matrices $\rho$ for $\cB$, which is
\ben\label{recov}
\tilde \alpha_{\eta,T}^t(\rho) = \sigma_\cA^{1/2-it} T\left( \sigma^{-1/2+it}_\cB \rho \sigma^{-1/2-it}_\cB \right) \sigma_\cA^{1/2+it}.
\een
An equivalent definition of the rotated Petz map is:

\begin{definition}
Let $T: \cB \to \cA$ be a unital, normal, and 2-positive, linear map and $|\eta_\cA\rangle \in \sH$ a faithful state. Then
the rotated Petz map $\alpha_{\eta,T}^t: \cA \to \cB$ is defined implicitly by the identity:
\ben
\label{eq:Petzdef}
\langle b \eta_{\cB} | J_\cB \Delta_{\eta_\cB}^{it} \alpha_{\eta,T}^t(a) \eta_{\cB} \rangle = 
\langle T(b) \eta_\cA | J_\cA \Delta_{\eta_\cA}^{it} a \eta_{\cA} \rangle,
\een
for all $a \in \cA, b \in \cB$.
\end{definition}

Closely related to the Petz map is the linear map $V_{\psi}: \sK \to \sH$  defined\footnote{
As it stands, the definition is actually consistent only when $|\xi_{\psi}^{\cB}\rangle$ is cyclic and separating. 
In the general case, one can define \cite{Petz4} instead
\ben
\label{eq:Vdef1}
V_{\psi}(b |\xi_{\psi}^{\cB}\rangle + |\zeta\rangle) := T(b) |\xi_{\psi}^{\cB} \rangle \quad (b \in \cB,
\pi^{\cB'}(\psi) |\zeta\rangle =0).
\een
} 
$\omega_\psi$ by \cite{Petz4,Petz2}
\ben
\label{eq:Vdef}
V_{\psi}b|\xi_{\psi}^{\cB}\rangle := T(b) |\xi_{\psi}^{\cA} \rangle \quad (b \in \cB).
\een
It follows from Kadison's property $T(a^*a) \ge T(a^*)T(a)$ (which is a consequence of \eqref{eq:eqschwarz}) that $V_{\psi}$ is a contraction $\|V_{\psi}\| \le 1$, see e.g. \cite{Petz4}, proof of thm. 4. 

As in paper II, we introduce a vector valued function 
\ben\label{eq:Gupper}
z \mapsto |\Gamma_\psi(z)\rangle := 
\Delta_{\eta_\cA,\psi_\cA}^{z} V_{\psi} \Delta_{\eta_\cB,\psi_\cB}^{-z} |\xi_{\psi}^{\cB}\rangle \quad (z\in \overline \bS_{1/2}), 
\een
the existence and properties of which are established in lem.s 3, 4 in paper II. In particular, $|\Gamma_\psi(z)\rangle$ is holomorphic inside 
the strip $\bS_{1/2}$ and bounded in the closure $\overline \bS_{1/2}$ in norm by 1. Furthermore, the representation (24) of paper I
shows in conjunction with Stone's theorem that this function is strongly continuous on the boundaries of the strip
$\bS_{1/2}$, i.e. for $\Re(z)=0$ or $\Re(z)=1/2$, which is used implicitly below e.g. when we consider integrals involving this quantity 
along these boundaries. The relation to the Petz map is 
as follows, paper II, lem. 2:
\ben\label{eq:pb}
\langle \Gamma_\psi(1/2+it) | a \, \Gamma_\psi(1/2+it) \rangle \le \omega_\psi \circ T \circ \alpha^t_{\eta,T}(a) \quad t \in \RR, a \in \cA_+.
\een
 
\subsection{Improved DPI}

Our main theorem is:

\begin{theorem}\label{thm:main}
Let $T:\cB \to \cA$ be a two-positive, unital (in the sense \eqref{eq:eqschwarz}) linear map between two von Neumann algebras, and let $\omega_\psi, \omega_\eta$ be normal states on $\cA$, with $\omega_\eta$ faithful. Then 
\ben
\label{mainth}
S(\omega_\psi | \omega_\eta) - S(\omega_\psi \circ T| \omega_\eta \circ T) \ge S_{\rm meas}(\omega_\psi | \omega_\psi \circ T \circ \alpha_{T,\eta}).
\een
with the recovery channel
\ben
\alpha_{T,\eta} \equiv \int_\RR dt \, \beta_0(t) \, \alpha^t_{T,\eta}.
\een
\end{theorem}
\noindent
{\bf Remarks:}
1) The theorem should generalize to non-faithful $\omega_\eta$ by applying appropriate support projections in a similar way as in paper I, lem. 1.

2) For finite-dimensional type I von Neumann algebras i.e. matrices, our result is due to \cite{Sutter2}. The recovery channel is given 
explicitly by \eqref{recov} in this case as an operator on density matrices, 
where $\sigma_\cA, \sigma_\cB$ are the density matrices corresponding to $\omega_\eta, \omega_\eta \circ T$.

3) By prop. \ref{FS}, our bound implies that given in our previous paper II for the fidelity; in fact it is stronger in many cases.

\medskip

{\bf I) Proof under a majorization condition:} First we consider the special case where there exists $\infty>c\ge 1$ such that 
\ben\label{relbound}
c^{-1} \omega_\eta \le \omega_\psi \le c\omega_\eta.
\een
Note that this implies $c^{-1} \omega_\eta \circ T \le \omega_\psi \circ T \le c\omega_\eta \circ T$ as $T$ is positive. 
By \cite{Petz1993}, thm. 12.11 (due to Araki), there exists a $h = h^* \in \cA$ such that $|\psi\rangle = |\eta^h\rangle/\|\eta^h\|$ such 
that $\|h\| \le \log c$, and vice versa. As is well known, this furthermore implies that the Connes cocycle $[D\eta^\cB:D\psi^\cB]_{iz}$
is holomorphic in the two-sided strip $\{z \in \CC : |\Re(z)|<1/2\}$ and bounded in norm (by $c^{\Re(z)}$) on the closure of this strip, see e.g. 
paper II, lem. 5. As a consequence, 
we have an absolutely convergent (in the operator norm) power series expansion 
\ben\label{power}
[D\eta^\cB:D\psi^\cB]_{iz} = 1+ \sum_{l=1}^\infty z^l k_l, 
\een
with bounded operators $k_l \in \cB$ such that $\|k_l\| \le C^l$. We set  
\ben\label{kdef}
k := \frac{d}{idt} T([D\eta^\cB:D\psi^\cB]_{t}) |_{t=0} \in \cA_{\rm s.a.} .
\een 
Using \cite{Petz1993}, cor. 12.8, and the definition of the relative entropy in terms of the Connes cocycle,
\ben
\begin{split}
S_\cA(\psi | \eta^k) =& S_\cA(\psi | \eta) - \omega_\psi(k) \\
=&  S_\cA(\psi | \eta) - \langle \psi^\cA | \frac{d}{idt} T([D\eta^\cB:D\psi^\cB]_{t}) \psi^\cA \rangle |_{t=0}\\
=& S_\cA(\psi | \eta) - S_\cB(\psi | \eta),
\end{split}
\een
which is one side of the inequality that we would like to prove. The variational expression \eqref{eq:var} then gives:
\ben
\label{eq:var2}
S_\cA(\psi | \eta) - S_\cB(\psi | \eta) = \sup_{h \in \cA_{\rm s.a.}} \{ \omega_\psi(h) - \log \| \eta^{h+k} \|^2 \}, 
\een
where we used $|(\eta^k)^h \rangle= |\eta^{k+h}\rangle$ see \cite{Petz1993}, thm. 12.10. To get the desired DPI we will establish an upper 
bound on $\log \| \eta^{h+k} \|^2$. 

In lem. \ref{lem:hirsch}, we take $|G(z)\rangle = e^{zh} |\Gamma_\psi(z)\rangle$, $p_0=\infty, p_1=2$ where $\theta=1/n$ with $n \in 4\mathbb N$
and $h=h^* \in \cA$. At the lower boundary we have with $u_\cB(t) := [D\eta^\cB:D\psi^\cB]_{t} \in \cB, u_\cA(t) := [D\eta^\cA:D\psi^\cA]_{t} \in \cA$ 
the unitary Connes cocycles,
\ben
\begin{split}
\|G(it)\|_{p_0,\psi} 
=& \| e^{ith} \Delta_{\eta_\cA,\psi_\cA}^{it} V_{\psi} \Delta_{\eta_\cB,\psi_\cB}^{-it} \xi_{\psi}^{\cB} \|_{\infty,\psi} \\
=& \| e^{ith} \Delta_{\eta_\cA,\psi_\cA}^{it} T(u_\cB(t)) \psi \|_{\infty,\psi}\\
=& \| e^{ith} \varsigma^t_\eta[T(u_\cB(t))] u_\cA(t)^*  \psi \|_{\infty,\psi}\\
=& \| e^{ith} \varsigma^t_\eta[T(u_\cB(t))] u_\cA(t)^*\|\\
= &  \| \varsigma^t_\eta[T(u_\cB(t))]\| \le 1,
\end{split}
\een
where we used 
$ \| \varsigma^t_\eta [T(b)] \| = \| T(b) \| \le \| b \| $
(from the positivity of $T$ and $\varsigma^t_\eta = {\rm Ad} \Delta^{it}_{\eta_\cA}$) as well as 
 the isomeric identification of $L^\infty(\cA, \psi) \owns a|\psi\rangle \mapsto a \in \cA$
proven in \cite{AM}.
Since $p_\theta = n$ and $\log \|G(it)\|_{p_0,\psi} \le 0$ as just shown, we get from lem. \ref{lem:hirsch}
\ben
\begin{split}
\log \| e^{h/n} \Gamma_\psi(1/n) \|_{\psi,n}^n \le& \int_\RR dt \, \beta_{1/n}(t) \, \log \| G(1/2+it)\|_{p_1,\psi} \\
=& \int_\RR dt \, \beta_{1/n}(t) \, \log \| e^{h/2} \Gamma_\psi(1/2+it)\|^2 \\
\le& \log \int_\RR dt \, \beta_{1/n}(t) \, \| e^{h/2} \Gamma_\psi(1/2+it)\|^2 \\
\le& \log \int_\RR dt \, \beta_{1/n}(t) \,  \omega_\psi \circ T \circ \alpha^t_{\eta,T}(e^h),
\end{split}
\een
using \eqref{eq:pb} in the third line and Jensen's inequality in the second (noting that the integrand is continuous and uniformly bounded).
Taking the lim-sup $n \to \infty$, we get using the definition of the recovery channel $\alpha_{T,\eta}$:
\ben\label{eq:var1}
\limsup_n \log \| e^{h/n} \Gamma_\psi(1/n) \|_{\psi,n}^n \le \omega_\psi \circ T \circ \alpha_{\eta,T}(e^h).
\een
The next lemmas give an expression for the lim-sup:
\begin{lemma}
We have $\| e^{h/n} \Gamma_{\psi}(1/n) \|_{\psi,n}^n = \|(e^{h/n} \Delta_{\eta,\psi}^{1/n} a_n \Delta_{\eta,\psi}^{1/n} e^{h/n} )^{n/4}_{}  \psi \|^2$, 
where 
\ben
a_n = T([D\eta^\cB:D\psi^\cB]_{i/n})^*T([D\eta^\cB:D\psi^\cB]_{i/n})\in \cA_+.
\een
\end{lemma}

\begin{lemma}
We have $\lim_n \|(e^{h/n} \Delta_{\eta,\psi}^{1/n} a_n \Delta_{\eta,\psi}^{1/n} e^{h/n} )^{n/4}_{}  \psi \|^2 = \|\eta^{h+k}\|^2$.
\end{lemma}

Combining the two lemmas with eq.s \eqref{eq:var2}, \eqref{eq:var1} gives
\ben
\label{eq:var}
S_\cA(\psi | \eta) - S_\cB(\psi | \eta) \ge 
\sup_{h \in \cA_{\rm s.a.}} \{ \omega_\psi(h) - \log \omega_\psi \circ T \circ \alpha_{\eta,T}(e^h) \}
=
S_{\rm meas}(\omega_\psi | \omega_\psi \circ T \circ \alpha_{T,\eta}),
\een
using the variational definition \eqref{SMdef} of $S_{\rm meas}$ in the last step.

\medskip
\noindent
{\it Proof of lem. 6:} Since \eqref{power} is an absolutely convergent power series in the operator norm, 
it follows that 
$
a_n = 1+2n^{-1}k + O(n^{-2})
$
where $O(n^\alpha)$ denotes a family of operators such that 
$\|O(n^\alpha)\| \le cn^\alpha$ for all $n>0$. Since $h$ is bounded, we also have $e^{h/n} = 1+n^{-1}h + O(n^{-2})$.  
Replacing $n \to 2n$ to simplify some expressions we trivially get
\ben
e^{h/(2n)} \Delta_{\eta,\psi}^{1/(2n)} a_{2n} \Delta_{\eta,\psi}^{1/(2n)} e^{h/(2n)}
=  \Delta_{\eta,\psi}^{1/n} + n^{-1}X_n + n^{-2} Y_n
\een
where $X_n,Y_n$ is a finite sum of terms of the form $x_{0} \Delta_{\eta,\psi}^{s_1} x_1 \cdots x_l \Delta_{\eta,\psi}^{s_l} x_{l}$ wherein 
$\sum s_j= 1/n, s_j\ge 0$ and $\|x_j\| \le c$ uniformly in $n$. $X_n$ is given explicitly by
\ben
X_n= \tfrac{1}{2} h \Delta_{\eta,\psi}^{1/n} + \tfrac{1}{2} \Delta_{\eta,\psi}^{1/n} h 
+ \Delta_{\eta,\psi}^{1/(2n)}k\Delta_{\eta,\psi}^{1/(2n)}.
\een
By \cite{Araki2}, II, proof of thm. 3.1, the functions 
\ben
F(z):=x_1 \Delta_{\eta,\psi}^{z_1} x_2 \cdots x_j \Delta_{\eta,\psi}^{z_j} x_{j+1} |\psi\rangle, \quad z \in \bar \bS_{1/2}^j
\een
defined for given $x_j \in \cA$ are (strongly) analytic in the domain $\bS_{1/2}^j := \{(z_1, \dots, z_j) \in \CC^j :  0 <\Re(z_i), \sum \Re(z_i) < 1/2\}$
and strongly continuous on the closure.  Subharmonic analysis as in \cite{Araki2}, II, proof of thm. 3.1, or \cite{AM} furthermore gives the bound
\ben
\label{bound}
\|F(z)\| \le 
\prod_i \| x_i \| , \quad \forall z \in \bar \bS_{1/2}^j.
\een
This bound, and the elementary formula 
\ben
(A+tB)^N=\sum_{j=0}^N t^j
\sum_{{\tiny
\begin{matrix}
m_0+...+m_j=N-j, \\
m_j \in \bN_0
\end{matrix}
}} A^{m_0} B
\cdots A^{m_{j-1}} B A^{m_j} ,
\een
shows that the difference
\ben
\begin{split}
|\zeta_n\rangle =& \ (e^{h/(2n)} \Delta_{\eta,\psi}^{1/(2n)} a_{2n} \Delta_{\eta,\psi}^{1/(2n)} e^{h/(2n)})^{n/2} |\psi\rangle\\
&-\sum_{j=0}^{n/2} \, n^{-j} \sum_{{\tiny
\begin{matrix}
m_0+...+m_j=n/2-j, \\
m_j \in \bN_0
\end{matrix}
}} \Delta_{\eta,\psi}^{m_0/n} X_n 
\cdots \Delta_{\eta,\psi}^{m_{j-1}/n} X_n \Delta_{\eta,\psi}^{m_j/n} |\psi\rangle
\end{split}
\een
is bounded in norm by 
\ben
\| \zeta_n \| \le \left( 1+ n^{-1}(\|h\| + \|k\|) + n^{-2}c \right)^{n/2}-\left( 1+ n^{-1}(\|h\| + \|k\|)\right)^{n/2}
\een
for some $c<\infty$, hence it tends to zero in norm as $n \to \infty$. Setting now
\ben
|\phi_{n,j}\rangle = n^{-j} \sum_{{\tiny
\begin{matrix}
m_0+...+m_j=n/2-j, \\
m_j \in \bN_0
\end{matrix}
}} \Delta_{\eta,\psi}^{m_0/n} X_n 
\cdots \Delta_{\eta,\psi}^{m_{j-1}/n} X_n \Delta_{\eta,\psi}^{m_j/n} |\psi\rangle ,
\een
the strong continuity of the functions $F$ and the usual definition of the Riemann integral implies
\ben
\begin{split}
&|\phi_j\rangle := \lim_n |\phi_{n,j}\rangle \\
=& \int_0^{1/2} ds_0 \dots \int_0^{s_{j-1}} ds_j \, \Delta_{\eta,\psi}^{s_0-s_1}(h+k)\Delta^{s_{1}-s_2}_{\eta,\psi} (h+k) 
\dots \Delta^{s_{j-1}-s_j}_{\eta,\psi} (h+k) \Delta^{s_j}_{\eta,\psi} |\psi\rangle, 
\end{split}
\een
and the usual perturbation theory by bounded operators as in \cite{Araki5}, prop. 16 
gives $\sum_{j=0}^\infty |\phi_j\rangle = e^{(\log \Delta_{\eta,\psi} + h + k)/2}| \psi \rangle$. Hence, 
\ben
\lim_n
(e^{h/(2n)} \Delta_{\eta,\psi}^{1/(2n)} a_{2n} \Delta_{\eta,\psi}^{1/(2n)} e^{h/(2n)})^{n/2} |\psi\rangle = e^{(\log \Delta_{\eta,\psi} + h + k)/2}| \psi \rangle
\een
strongly, as argued more carefully in \cite{Araki4}, proof of lem. 5.
We have $e^{(\log \Delta_{\eta,\psi} + h + k)/2}| \psi \rangle= e^{(\log \Delta_{\eta,\psi} + p'h + p'k)/2}| \psi \rangle$
(here $p' = \pi^{\cA'}(\psi) \in \cA'$). 
Also, using \cite{Petz1993}, thm. 12.6., we have $\log \Delta_{\eta,\psi} + p'h + p'k = \log \Delta_{\eta^{h+k},\psi}$, and this 
gives $| \eta^{h+k}\rangle = J| \eta^{h+k}\rangle = e^{(\log \Delta_{\eta,\psi} + h + k)/2}| \psi \rangle$ by 
relative modular theory. This completes the proof. \qed

\medskip
\noindent
{\it Proof of lem. 5:} From the definitions, 
\ben
e^{h/n} \Gamma_\psi(1/n)  
= e^{h/n} \Delta_{\eta_\cA, \psi_\cA}^{1/n} V_\psi \Delta_{\eta_\cB, \psi_\cB}^{-1/n} \Delta_{\psi_\cB}^{1/n}  |\psi^\cB\rangle 
= e^{h/n} \Delta_{\eta_\cA, \psi_\cA}^{1/n}  T([D\eta^\cB:D\psi^\cB]_{i/n}) |\psi^\cA\rangle, 
\een
using the definition of the Connes-cocylce and the fact that $[D\eta^\cB:D\psi^\cB]_{i/n} \in \cB$ under our assumption \eqref{relbound}, 
see paper II, proof of lem. 4. In the following, let $a=e^{h/n}, b=T([D\eta^\cB:D\psi^\cB]_{i/n}) \in \cA$ and $|\psi^\cA\rangle=|\psi\rangle, 
|\eta^\cA\rangle=|\eta\rangle$ etc.

By the results of \cite{AM} (which hold in the present context since $\omega_\psi$ is faithful being dominated by the 
faithful state $\omega_\eta$), 
the vector $b \Delta_{\eta,\psi}^{1/n} a |\psi\rangle \in L_n(\cA,\psi)$ has a polar decomposition 
$b \Delta_{\eta,\psi}^{1/n} a |\psi\rangle = u\Delta_{\phi_n,\psi}^{1/n} |\psi\rangle$, where $\| b \Delta_{\eta,\psi}^{1/n} a \psi \|^n_{n,\psi} = \|\phi_n\|^2$
and where $u \in \cA$ is a partial isometry. To get an expression for $|\phi_n\rangle$, we use the formalism of ``script'' $\sL_p$-spaces of 
\cite{AM}, notation 7.6: As a vector space $\sL_p^*(\cA,\psi), p \ge 1$ consists of all formal linear combinations of formal expressions of the form
\ben
A=x_1 \Delta^{z_1}_{\zeta_1,\psi} x_2 \dots x_{n}  \Delta^{z_n}_{\zeta_n,\psi} x_{n+1}
\een
wherein $\Re(z_i) \ge 0, \sum_i \Re(z_i) \le 1-1/p$, $x_i \in \cA, \zeta_i \in \sH$, the formal adjoint of which is defined to be 
\ben
A^*=x_{n+1}^* \Delta^{\bar z_n}_{\zeta_n,\psi} x_n^* \dots x_{2}^*  \Delta^{\bar z_1}_{\zeta_1,\psi} x_{1}^*.
\een
The notation $\sL_{p,0}^*(\cA,\psi)$ is reserved for formal elements $A$ such that $\sum_i \Re(z_i) = 1-1/p$ in addition to all other 
conditions. It is then clear that $\sL_{p,0}^*(\cA,\psi) \sL_{q,0}^*(\cM,\psi) = \sL_{r,0}^*(\cM,\psi)$ as formal products where 
$1/r' = 1/p' + 1/q'$ with $1/p'=1-1/p$ as usual. By \cite{AM}, lem. 7.3, if $1 \le p \le 2$, any element $A \in \sL_{p}^*(\cA,\psi)$ can be viewed as an element of $L_{p'}(\cA,\psi)$ in the sense that $|\psi\rangle \in \sD(A)$ and $A|\psi\rangle \in L_{p'}(\cA,\psi)$.\footnote{
In fact, $\| A\psi\|_{p',\psi} \le \| x_{n+1} \| \prod_{i=1}^n (\| x_i \| \| \zeta_i \|^{\Re(z_i)})$.
} Furthermore, by \cite{AM}, lem. 7.7 (2), if $A_1, A_2 \in \sL_p^*(\cA,\psi)$
correspond to the same element under this identification, then so do $A_1^*, A^*_2$ or $A_1B, A_2B$ or $BA_1, BA_2$ if $B \in \sL_{q,0}^*(\cA,\psi)$ (as long as 
$1/p'+1/q' \le 1/2$, for example). 

We now start with the trivial statement that $u\Delta_{\phi_n,\psi}^{1/n} = b \Delta_{\eta,\psi}^{1/n} a$ in the sense that these elements of $\sL_{n',0}^*(\cA,\psi)$ are identified with the same element of $L_n(\cA,\psi)$. Then repeated application of \cite{AM}, lem. 7.7 (2) and the definition of adjoint gives 
\ben
u\Delta_{\phi_n,\psi}^{2/n} u^* = b \Delta_{\eta,\psi}^{1/n} aa^* \Delta_{\eta,\psi}^{1/n} b^* \quad \text{in $\sL_{n/(n-2),0}^*(\cA,\psi)$.}
\een
Forming successively $n/4$ products of this equality and applying in each step \cite{AM}, lem. 7.7 (2), we find that 
\ben
u\Delta_{\phi_n,\psi}^{1/2} u^* = (b \Delta_{\eta,\psi}^{1/n} aa^* \Delta_{\eta,\psi}^{1/n} b^*)_{}^{n/4} \quad \text{in $\sL_{2,0}^*(\cA,\psi)$,}
\een
meaning that both sides are equal as elements of $\sH=L_2(\cA,\psi)$ after we apply them to $|\psi\rangle$. Thus,  
\ben
\|(b \Delta_{\eta,\psi}^{1/n} aa^* \Delta_{\eta,\psi}^{1/n} b^*)_{}^{n/4} \psi\|^2 = \|u\Delta_{\phi_n,\psi}^{1/2} u^*\psi\|^2 = 
\|uJu\phi_n\|^2 = \|\phi_n\|^2
\een
using modular theory. Therefore 
\ben
\|(b \Delta_{\eta,\psi}^{1/n} aa^* \Delta_{\eta,\psi}^{1/n} b^*)_{}^{n/4} \psi\|^2 = \| b \Delta_{\eta,\psi}^{1/n} a \psi \|^n_{n,\psi},
\een
and the proof of the lemma
is complete. \qed

\medskip
\noindent
{\bf II) Proof in general case:} We will now remove the majorization condition \eqref{relbound}. This condition has been used in an essential way in most of the arguments so far because without it, the operator $k$ in \eqref{kdef} is unbounded and thus not an element of $\cA$. For unbounded operators the Araki-Trotter product formula and the $L_p$-techniques are not available and it seems non-trivial extending them to an unbounded framework. We will therefore proceed in a different way and define a regularization of $\omega_\psi$ such that the majorization condition \eqref{relbound} holds and such that, at the same time, the desired entropy inequality can be obtained in a limit wherein the regulator is removed. However, it is clear that this regularization must be carefully chosen because the relative entropy is not continuous but only lower semi-continuous. By itself the latter is insufficient for our purposes since the desired inequality \eqref{mainth} has both signs of the relative entropy. 

Our regularization combines a trick invented in paper I with the convexity of the relative entropy. 
As in paper I, we consider a function $f(t), t \in \RR$ with the following properties. 
\begin{itemize}
\item[(A)] The Fourier transform of $f$ 
\begin{equation}
\tilde{f}(p) = \int_{-\infty}^{\infty}  e^{ - i t p} f(t) dt
\end{equation}
exists as a real and non-negative Schwarz-space function.
This implies that the original function $f$ is Schwarz and has finite $L_1(\mathbb{R})$ norm,  $\| f \|_1 < \infty$.
\item[(B)] $f(t)$ has an analytic continuation to the upper complex half plane such that  the $L_1(\mathbb{R})$ norm 
of the shifted function has $\| f( {\cdot} + i\theta) \|_1 < \infty$
for $0 < \theta < \infty$. 
\end{itemize}
Such functions certainly exist (e.g. Gaussians). 
We also let $f_P(t) = P f( t P)$ for our regulator $P>0$, and we define a regulated version of $|\psi\rangle$ by 
\ben
\left| \psi_{P} \right> = \frac{\tilde{f}_P(\ln \Delta_{\eta,\psi} ) \left| \psi \right> }{\|\tilde{f}_P(\ln \Delta_{\eta,\psi} )\psi\| }. 
\een
As shown in paper I, some key properties of the regulated vectors are:
\begin{enumerate} 
\item[(P1)]
$
\omega_{\psi_P} \le c_P \omega_\eta 
$
for some $c_P >0$ which may diverge as $P \to \infty$,
\item[(P2)] 
$s-\lim_{P \to \infty} |\psi_P\rangle = |\psi\rangle$ (strong convergence),
\item[(P3)] 
$- 2 \ln\left(  \|f \|_1/\| \tilde f \|_\infty  \right) 
+ \limsup_{P \rightarrow \infty} S(\psi_{P}|\eta) 
\leq S(\psi|\eta)$,
\end{enumerate}
where the first item gives at least ``half'' of the domination condition \eqref{relbound}, the second states in which sense $|\psi_P\rangle$ 
approximates $|\psi\rangle$ and the third gives us an upper semi-continuity property of the relative entropy opposite to the usual lower semi-continuity property which holds for generic approximations. We define for small $\epsilon > 0$:
\ben
\sigma(a) = \langle \eta | a \eta \rangle \quad 
\rho_{P,\epsilon}(a) = (1-\epsilon) \langle \psi_P | a\psi_P \rangle + \epsilon \langle \eta | a\eta \rangle .
\een
Thus, by P1), the relative majorization condition \eqref{relbound} holds e.g. with $c={\rm max}(c_P,\epsilon^{-1})$  
between $\rho_{P,\epsilon}$ and $\sigma$.  
By P2), $\lim_{P \to \infty} \lim_{\epsilon \to 0} \| \rho-\rho_{P,\epsilon} \|=0$. In P3), we choose a function $f$ 
such that  $\|f \|_1/\| \tilde f \|_\infty=1$ (which must be Gaussian). The well-known convexity of the relative entropy gives
together with the definition of $\rho_{\epsilon,P}$ that ($\rho_P = \langle \psi_P | \ . \ \psi_P \rangle$)
\ben
S(\rho_{P ,\epsilon} | \sigma) \le (1-\epsilon) \ S(\rho_P,\sigma) + \epsilon \ S(\sigma | \sigma)= (1-\epsilon) \, S(\rho_P,\sigma).
\een
Combining this with P3), we get
\ben\label{x1}
\limsup_{P \to \infty} \limsup_{\epsilon \to 0} S(\rho_{P,\epsilon} | \sigma) \le S(\rho | \sigma). 
\een
The norm convergence $\lim_P \lim_\epsilon \rho_{P,\epsilon} \circ T = \rho \circ T$ by P2) also gives in combination with the usual lower
semi-continuity of the relative entropy, \cite{Araki2}, II thm. 3.7 (2), that 
\ben\label{x2}
\liminf_{P\to \infty} \liminf_{\epsilon \to 0} S(\rho_{P,\epsilon} \circ T | \sigma \circ T) \ge S(\rho \circ T| \sigma \circ T). 
\een
Now we combine eq.s \eqref{x1}, \eqref{x2} with 
part I of the proof applied to the states $\rho_{P,\epsilon}$ and $\sigma$, which obey the relative majorization condition. 
We get:
\ben
S(\rho | \sigma) - S(\rho \circ T| \sigma \circ T) \ge \limsup_{P \to \infty} \limsup_{\epsilon \to 0}
S_{\rm meas}(\rho_{P,\epsilon} | \rho_{P,\epsilon} \circ T \circ \alpha_{T,\sigma}).
\een
The proof of part II is then finished by proving lower semi-continuity for the measured relative entropy:
\begin{lemma}
If $\mu_n, \nu_n, \mu, \nu \in \sS(\cA)$ are such that $\lim_n \mu_n = \mu$ and $\lim_n \nu_n = \nu$ in the norm sense, then
$S_{\rm meas}(\mu | \nu) \le \liminf_n S_{\rm meas}(\mu_n | \nu_n)$.
\end{lemma}

\begin{proof}
This is a straightforward consequence of the variational definition \eqref{SMdef} of $S_{\rm meas}$, choosing a near optimal $h$. 
\end{proof}

\qed

{\bf Acknowledgements:} SH\ thanks Tom Faulkner for conversations and the Max-Planck Society for supporting the collaboration between MPI-MiS and Leipzig U., grant Proj.~Bez.\ M.FE.A.MATN0003.

\appendix

\section{Weighted $L_p$ spaces \cite{AM} and variational formulae}

The weighted $L_p$-spaces were defined by \cite{AM} relative to a fixed vector $|\psi\rangle \in \H$ in the a natural cone of a standard representation of a von Neumann algebra $\cM$. 
For $p \ge 2$, the space $L_p(\cM, \psi)$ is defined as 
\ben
L_p(\cM, \psi) = \{ |\zeta\rangle \in \bigcap_{|\phi\rangle \in \H} \sD(\Delta_{\phi, \psi}^{(1/2)-(1/p)}), \|\zeta\|_{p,\psi} < \infty\}.
\een
Here, the norm is 
\ben
\|\zeta\|_{p,\psi} = \sup_{\|\phi\|=1} \|\Delta_{\phi, \psi}^{(1/2)-(1/p)}\zeta\|.
\een
For $1 \le p < 2$, $L_p(\cM, \psi)$ is defined as the completion of $\sH$ with respect to the following norm:
\ben\label{eq:pnorm}
\|\zeta\|_{p,\psi} = \inf \{  \|\Delta_{\phi, \psi}^{(1/2)-(1/p)}\zeta\| : \|\phi\|=1, \pi^\cM(\phi) \ge \pi^\cM(\psi) \}.
\een
In \cite{AM}, it is assumed for most results that $|\psi\rangle$ is cyclic and separating. When using such results in the main text, 
we will be in that situation. A somewhat different approach replacing the relative modular operator by the Connes spatial derivative
and containing also many new results is laid out some detail in \cite{Berta2}.

\end{document}